\documentclass{article}

\PassOptionsToPackage{numbers,compress}{natbib}
\usepackage[preprint]{nips_2018}

\usepackage[utf8]{inputenc} 
\usepackage[T1]{fontenc}    

\usepackage{url}            
\usepackage{booktabs}       
\usepackage{amsfonts}       
\usepackage{nicefrac}       
\usepackage{microtype}      
\usepackage{amsmath}
\usepackage{amsthm}
\usepackage{amssymb}
\usepackage{todonotes}
\usepackage{graphicx}
\usepackage{subfigure}
\graphicspath{ {../figures/chaos/} {../figures/traffic/} {../figures/gmm/} {../figures/}}
\usepackage{paralist}
\usepackage{enumerate}
\usepackage{enumitem}
\usepackage{epstopdf}
\usepackage{wrapfig}

\usepackage{marginnote}

\usepackage{thmtools}
\usepackage{thm-restate}

\theoremstyle{plain}

\newtheorem{lemma}{Lemma}
\newtheorem{corollary}{Corollary}
\theoremstyle{remark}
\newtheorem{remark}{Remark}
\theoremstyle{assumption}
\newtheorem{assumption}{Assumption}

\usepackage{dsfont}
\newcommand{\iden}{\ensuremath{\mathbb{I}}}  
\newcommand{\real}{\ensuremath{\mathbb{R}}} 
\newcommand{\expt}{\ensuremath{\mathbb{E}}}
\DeclareMathOperator*{\var}{Var}
\newcommand{\X}{\ensuremath{\mathbb{X}}}
\newcommand{\Z}{\ensuremath{\mathfrak{Z}}}

\newcommand{\pto}{\overset{p}{\to}}

\DeclareMathOperator*{\argmin}{arg\,\!min}

\usepackage{algorithmicx}
\usepackage{algorithm}
\usepackage[noend]{algpseudocode}
\usepackage{setspace}

\algnewcommand{\IIf}[1]{\State\algorithmicif\ #1\ \algorithmicthen}
\algnewcommand{\ElseIIf}{\unskip\ \algorithmicelse}
\algnewcommand{\EndIIf}{\unskip\ \algorithmicend\ \algorithmicif}

\usepackage{titlesec}
\titlespacing\section{0pt}{4pt plus 2pt minus 2pt}{0pt plus 2pt minus 0pt}
\titlespacing\subsection{0pt}{4pt plus 2pt minus 2pt}{0pt plus 2pt minus 0pt}
\titlespacing\subsubsection{0pt}{4pt plus 2pt minus 2pt}{0pt plus 2pt minus 0pt}

\title{Inference Trees: Adaptive Inference with Exploration}

\author{
    Tom Rainforth$^1$~~ Yuan Zhou$^1$~~ Xiaoyu Lu$^1$~~ Yee Whye Teh$^1$~~ Frank Wood$^2$ \\
    \textbf{Hongseok Yang}$^3$~~ \textbf{Jan-Willem van de Meent}$^4$ \\
    $^1$University of Oxford; ~
    $^2$University of British Columbia; ~
    $^3$KAIST; ~
    $^4$Northeastern University \\
    \small{\texttt{\{rainforth, xiaoyu.lu,  y.w.teh\}@stats.ox.ac.uk, yuan.zhou@cs.ox.ac.uk,}} \\
    \small{\texttt{fwood@cs.ubc.ca, hongseok.yang@kaist.ac.kr, j.vandemeent@northeastern.edu}}
    \vspace{-15pt}
}

\begin{document}

\maketitle
	
\setlength{\abovedisplayskip}{2.5pt}
\setlength{\belowdisplayskip}{2.5pt}
\setlength{\abovedisplayshortskip}{2.5pt}
\setlength{\belowdisplayshortskip}{2.5pt}	
	
\begin{abstract}
	\vspace{-5pt}

We introduce inference trees (ITs), a new class of inference methods
that build on ideas from Monte Carlo tree search to perform adaptive sampling in a manner that balances exploration with exploitation, ensures consistency, 
and alleviates pathologies in existing adaptive methods.
ITs adaptively sample from hierarchical partitions of the parameter space,
while simultaneously learning these partitions in an online manner. 
This enables ITs to not only identify regions of high posterior mass, but also 
maintain uncertainty estimates to track regions where significant posterior 
mass may have been missed. ITs can be based on any inference method 
that provides a consistent estimate of the marginal likelihood. They are 
particularly effective when combined with sequential Monte Carlo, where they  capture long-range dependencies and yield improvements beyond proposal adaptation alone.

\end{abstract}
	
\section{Introduction}
\label{sec:introduction}


The choice of proposal distribution is a key factor in the performance of Monte Carlo (MC) methods.
Unfortunately, it is typically difficult to know what constitutes a good 
proposal prior to performing inference. 
For this reason, many methods use past samples to adapt the proposal at future iterations \cite{cappe2004population,cornebise2008adaptive,
cornuet2012adaptive,gu2015neural,liang2011advanced}, for example by minimizing the KL divergence between the empirical distribution over samples and the proposal. These strategies implicitly assume that preceding samples are representative of the true posterior. This leads to the somewhat undesirable characteristic that we already need good samples to have effective adaptation, which is presumably difficult to achieve given our need to adapt in the first place. Adaptive methods can 
consequently exhibit pathologies, such as collapsing to a single mode or even adapting to invalid proposals~\cite{andrieu2008tutorial,cappe2008adaptive}.


To address these issues, we propose that adaptive methods should not only
carry out \emph{exploitation}, that is sample in regions where we believe the
posterior mass is high, but also \emph{exploration}, that is explicitly invest computational resources to
sample in regions where our current uncertainty about the posterior mass is
high. In other words, we should recognize that the utility derived from the generated samples 
originates not only from their direct contribution to the
estimator, but also the degree to which they inform future sampling.


To this end, we introduce
inference trees (ITs), a new class of adaptive methods that build on ideas from Monte Carlo tree search
(MCTS) \cite{browne2012survey,kocsis2006bandit}. ITs hierarchically partition the parameter space into disjoint regions in an online manner, resulting in more fine-grained partitions for regions where the posterior density is large. This transforms the problem of inference on the full parameter space to a set of constrained inference problems, which we can
combine in a manner akin to stratified sampling~\cite{carpentier2015adaptive,neufeld2014adaptive}. 
By adaptively choosing regions in which to refine our estimates, we can
explicitly control the exploration-exploitation trade-off.
This results in an algorithm that can expend computational resources to investigate
whether the proposal can be improved, for example by searching for missing modes, 
rather than just greedily exploiting the best proposal learned so far. 


ITs can be thought of as a meta-algorithm that controls the allocation of computational resources of a base inference algorithm.
We show that, under mild assumptions, ITs define a consistent estimator whenever the base algorithm itself provides a consistent estimator. This property is independent of the methods for learning the partitioning and allocation of
computational resources between the partitions.
In addition to the theoretical guarantees that this provides, the resulting flexibility
proves critical to the empirical performance of ITs.
For example, we exploit this flexibility to introduce a novel allocation scheme that uses \emph{targeted} exploration: rather than
just using an optimism boost~\cite{auer2002finite} to ensure a minimum level of allocation for all regions, it uses
explicit uncertainty estimates for the true marginal posterior mass of a region to identify important areas to explore, such as those
likely to contain a missing mode.  Underlying this approach is a novel estimator in its own right.  Namely, we
perform density estimation on sample weights to predict the probability the true marginal posterior mass of a region
is above a certain threshold.  Remarkably, this estimator remains robust even when the MC estimate of the marginal
is thousands of orders of magnitude smaller than the true value.

We find that the gains that ITs provide are particularly 
pronounced when they are combined with sequential Monte Carlo (SMC)~\cite{doucet2001introduction}, where they offer a means of capturing long-range dependencies. This yields improvements beyond what can be achieved by the so-called one-step optimal proposal.


\section{Background and Related Work}
\label{sec:background}

Our aim is to approximate a target density $\pi(x) = \gamma(x)/\omega$, for which it is possible to evaluate the unnormalized density $\gamma(x)$ pointwise, but computation of the normalization constant $\omega$ is intractable.
We will assume that we have a base MC
algorithm that returns weighted
samples and makes use of some form of proposal distribution $q(x)$.  Though we will later consider
other approaches (see \S\ref{sec:chaos}), for exposition, it will be easiest to think of this base
algorithm as being self-normalized importance 
sampling~\citep{owen2013mc}, which defines an estimated measure based on weighted samples from $q$,
\begin{align}
\hat{\pi} (\cdot) := \sum_{n=1}^{N} \bar{w}_n \delta_{\hat{x}^n}(\cdot) \quad
\text{where} \quad \hat{x}^n \sim q(x), \quad
 \bar{w}_n &:= \frac{w_n}{\sum_{n=1}^{N} w_n},
\quad w_n := \frac{\gamma(\hat{x}^n)}{q(\hat{x}^n)}.
\end{align}

\subsection{Adaptive Monte Carlo Inference}
\label{sec:adapt_MC}

Though there are a range of approaches for adapting the proposal $q$ 
(see \citet{bugallo2017adaptive} for a review),
most share a common framework of alternating between sampling using the
current proposal and adapting the proposal using previous samples,
with the latter often taking the form of a (potentially implicit) density estimation.
For example, one common approach is to, at each
iteration, choose the proposal that minimizes the KL divergence from the
estimated posterior to the proposal~\citep{cappe2008adaptive,douc2007convergence}.
Namely, if $\theta$ denotes the parameters of $q$, one uses 
$\theta^* = \argmin_{\theta} \sum_{n=1}^N -\bar{w}_n \log q_{\theta} (\hat{x}^n)$
at each iteration.
This leads to an expectation maximization style approach
that is \emph{greedy}, in the sense that past samples are assumed to accurately represent the posterior.

\subsection{Multi-Armed Bandits and Monte Carlo Tree Search} 
\label{sec:bandits}

In multi-armed bandit problems, an agent sequentially chooses between multiple actions,
known as arms, each of which returns a stochastic reward.  The agent's goal is to 
maximize the long-term cumulative reward~\citep{agrawal2012analysis,berry1985bandit}.  
One common strategy is upper confidence bounding (UCB)~\citep{auer2002finite}, which chooses the arm $j$ that maximizes the utility 
\begin{align}
\label{eq:ucb}
u_j = \hat{r}_j + (\beta/\sqrt{M_j})\log \textstyle\sum\nolimits_i M_i.
\end{align}
In this definition, $\hat{r}_j \in [0,1]$ is the current estimate of the expected reward for each arm, $M_j$ is the number of times arm $j$ was previously pulled, and $\beta$ is a parameter
that controls the level of exploration. Here $\hat{r}_j$ is an \emph{exploitation} term that ensures we pull arms with high expected reward more frequently,
while $(\beta/\sqrt{M_j})\log \textstyle\sum\nolimits_i M_i$ is an \emph{exploration} term,
sometimes known as an optimism boost, which encourages us to pull arms which have been pulled infrequently so far.

Of particular relevance to our work is the study of bandits in the stratified sampling 
setting~\cite{carpentier2015adaptive,etore2010adaptive,etore2011adaptive,kawai2010asymptotically,
	neufeld2016adaptive}.
Here one splits a target integral into a number of strata, then looks to minimize the overall error by
allocating samples to the MC estimators associated with each strata.
The optimal strategy can
be shown to sample each strata in proportion to the standard deviation of its evaluations~\citep{carpentier2015adaptive}.
Because one now needs to asymptotically sample from each arm infinitely often,
  the strategy is adjusted to
\begin{align}
\label{eq:ucb-str}
u_j = \left(\hat{r}_j + (\beta/\sqrt{M_j})\log \textstyle\sum\nolimits_i M_i\right)/M_j
\end{align}
where $\hat{r}_j$ is typically set to the empirical standard deviation.  ITs differ from these
approaches in that they use hierarchical stratification, learn this stratification
in an online manner, use a different
utility that incorporates a targeted exploration term, and 
adapt UCB to the inference setting.

%

MCTS \cite{browne2012survey,kocsis2006bandit} uses a
hierarchy of arms where one traverses the tree by sequentially choosing child nodes using~\eqref{eq:ucb} (or a variation
thereof)
until a leaf is reached, then refines that node and propagates the new
estimates up through the tree.  The average reward $\hat{r}_j$ of a non-leaf node
is the average of those of its children, while $M_j$ becomes the number of times a node has been traversed.
MCTS has traditionally been used 
for planning~\citep{silver2016mastering} and in discrete decision settings.  
We believe that our work is the first to consider MCTS in the context of inference or integration, as opposed to optimization.  
ITs also vary from the standard MCTS setting in how rewards are calculated and propagated.

\section{Algorithm Overview}
\label{sec:overview}


ITs hierarchically partition the target space, run inference separately on the
resulting disjoint regions to obtain local estimates, and then combine these local
estimates into one overall estimate.  Each node in the tree corresponds to a region of 
target space, $A_j$, such that the region of a parent node is the union of its children, 
$A_j = A_{\ell_j} \cup A_{r_j}$ where $\ell_j$ and $r_j$ are the child indices, and the union of all leaf nodes is the full space. 
We assume that we are able to sample from the proposal restricted to a node, 
$q(x | x \in A_j)$, and evaluate this renormalized truncated density pointwise.
How this is achieved is discussed in \S\ref{sec:projection}.
 The IT learning process can be broken down into
three components as follows.

\begin{wrapfigure}{r}{0.29\textwidth}
	\vspace{-10pt}
\includegraphics[width=0.29\textwidth]{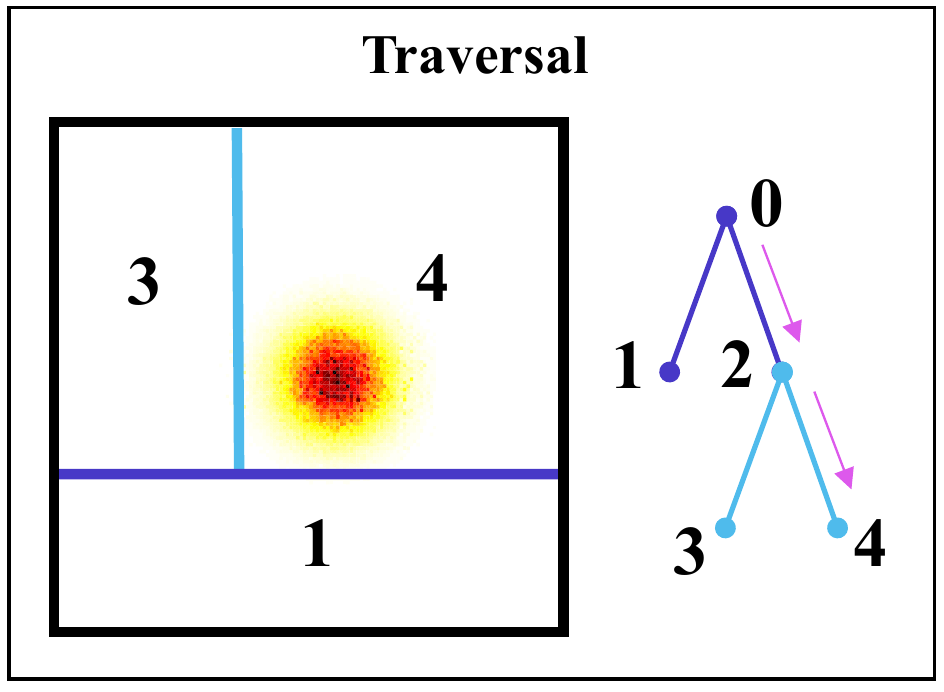}	

\vspace{5pt}
	\includegraphics[width=0.29\textwidth]{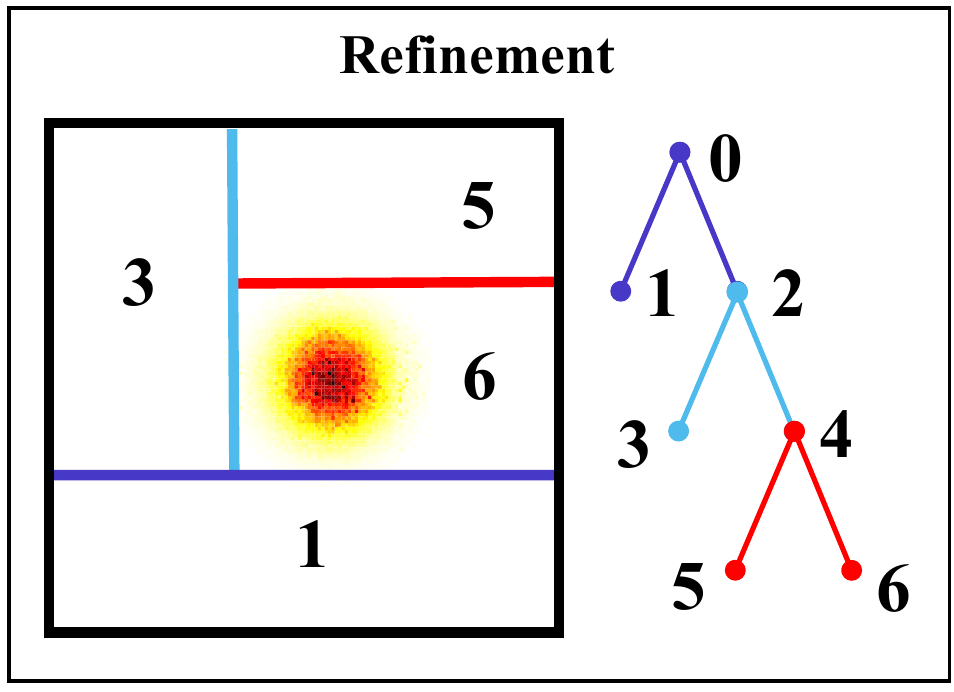}	
\vspace{-15pt}
\end{wrapfigure}
\textbf{Traversal}:
The traversal step 
adaptively allocates computational resources to areas of
the target space, balancing exploration and exploitation to
minimize the error of our final overall estimate.  
Following similar lines to MCTS, it
starts at the root node and then recursively choosing a child node until a leaf is reached.
To choose between children, we use the stratified sampling UCB formulation given
in~\eqref{eq:ucb-str}. 
However, as we explain in detail \S\ref{sec:traversal}, our $\hat{r}_j$ will vary from standard settings:
our reward must be adapted to reflect the fact we are doing inference
and rather than relying solely on the optimism boost for exploration, we will 
incorporate a targeted exploration term.

\textbf{Refinement}:
In the refinement step we improve the estimate at the chosen node, 
either by running inference directly and 
updating the local estimate, or expanding the tree by splitting the 
node and running inference at each of the generated child nodes.
For both cases, the inference itself is performed using the base
algorithm and the truncated proposal $q\left(x \mid x \in A_j\right)$. 
The two considerations for refinement are whether to split and how to split. They are
discussed in \S\ref{sec:refinement}.

\begin{wrapfigure}{r}{0.118\textwidth}
	\vspace{-20 pt}
	\includegraphics[width=0.12\textwidth]{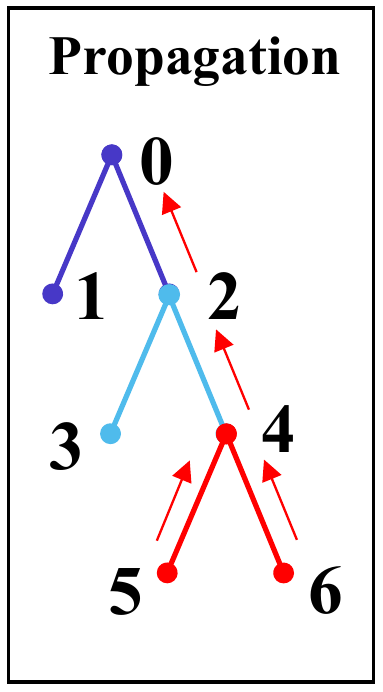}
	\vspace{-18pt}	
\end{wrapfigure}
\textbf{Propagation}:
In the propagation step, we recursively update the tree with the new estimates produced by
the refinement step, starting with the refined node(s) and then updating all their ancestors.
This improves our posterior representation and guides the future traversal strategy.
Along with a small number of additional terms required for the traversal, 
two key quantities are propagated up through the tree: a marginal likelihood estimate
$\hat{\omega}_j$ and an unnormalized empirical measure $\hat{\gamma}_j(\cdot)$.
The truncated posterior approximation at any node in the tree is then given by the 
self-normalized estimated measure $\hat{\pi}_j(\cdot) = \hat{\gamma}_j(\cdot) / \omega_j$,
with the root note estimate $\hat{\pi}_0(\cdot)$ representing our overall approximation.
The specifics of the propagation are discussed in \S\ref{sec:ITest}.

Putting these components together leads to an adaptive online inference algorithm
as summarized in Algorithm~\ref{alg:IT}.  We now discuss the individual elements of
ITs in more detail.  We note that, while the method for propagation is tightly coupled
with the IT estimator itself, the consistency of this estimator is independent of the
traversal and refinement strategies.  Consequently, a wide range of possible approaches fall under the general IT framework we have just introduced.

%

\setlength{\textfloatsep}{5pt}
\begin{algorithm}[t]
	\small
	\caption{Inference Tree Training \label{alg:IT}}
	\setstretch{1.1}
	\begin{algorithmic}[1]
		\renewcommand{\algorithmicrequire}{\textbf{Inputs:}}
		\renewcommand{\algorithmicensure}{\textbf{Outputs:}}	
		\Require Unnormalized target density $\gamma(x)$, ``truncatable'' proposal $q(x)$, base inference
		algorithm $\mathcal{F}$, complete target space $A_0$, number of iterations to run $R$, batch size $b$, existing tree $\mathcal{T}$ (optional)
		\Ensure Inference Tree $\mathcal{T}$,
		posterior empirical measure $\hat{\pi}_0(\cdot)$
		\State If required, initialize tree $\mathcal{T}$ by running inference on full space $\left\{\hat{x}_{0}^n,w_0^n\right\}_{n=1}^b \leftarrow \mathcal{F}\left(\gamma(x),q\left(x | A_0\right),b\right)$
		\For{$r = 1:R$} 
			\State Traverse tree by recursively selecting the child with highest
			$u_j$ (see~\eqref{eq:uct_it}) until a leaf $(j)$ is reached
			\If{decide to split node $j$} \Comment See \S\ref{sec:refinement}
			\State Use existing samples to split node $A_{\ell_j}, A_{r_j} \leftarrow A_j$ where
			$A_j = A_{\ell_j} \cup A_{r_j}$ \Comment See \S\ref{sec:projection} and \S\ref{sec:refinement}
			\State $\{\hat{x}_{\ell_j}^n,w_{\ell_j}^n\}_{n=1}^b \leftarrow \mathcal{F}\left(\gamma(x),q\left(x| A_{\ell_j}\right),b\right)$, \quad
			$\{\hat{x}_{r_j}^n,w_{r_j}^n\}_{n=1}^b \leftarrow \mathcal{F}\left(\gamma(x),q\left(x | A_{r_j}\right),b\right)$
			\Else {~} Run further inference on current node
			$\left\{\hat{x}_{j}^n,w_j^n\right\}_{n=N_j+1}^{N_j+b} \leftarrow \mathcal{F}\left(\gamma(x),q\left(x | A_j\right),b\right)$ ~ {\textbf{end if}} \vspace{-2pt}
			\EndIf
			\State Update $\hat{\gamma}_j(\cdot)$ and $\hat{\omega}_j$ for refined node(s) and all ancestors using~\eqref{eq:emp-measure-prop} and~\eqref{eq:ml-estimate} respectively ~ {\textbf{end for}}
		\EndFor
		\State Return tree $\mathcal{T}$ and self normalized empirical measure $\hat{\pi}_0(\cdot)\leftarrow\hat{\gamma}_0(\cdot)/\hat{\omega}_0$
	\end{algorithmic}
\end{algorithm}

\section{The Inference Tree Estimator}
\label{sec:ITest}

Assume we are trying to estimate the expectation
of a measurable function $f(x)$ with respect to the target measure $\pi(x)$.  
For any set of disjoint regions $\{A_i\}_{i \in \mathcal{I}}$
covering the full target space, we have
\begin{align}
\expt_{\pi(x)}[f(x)] &=
\frac{1}{\omega}\int \gamma(x) f(x) \, dx
= \frac{1}{\omega}\int \gamma(x) f(x) \sum_{i \in \mathcal{I}} \iden(x \in A_i) \, dx
= \frac{\sum_{i \in \mathcal{I}} \expt_{q\left(x \mid A_i\right)} 
\left[\frac{ \gamma(x) f(x) }
{q\left(x \mid A_i\right)}\right]}
{\sum_{i \in \mathcal{I}} \expt_{q\left(x \mid A_i\right)} 
	\left[\frac{ \gamma(x)}
	{q\left(x \mid A_i\right)}\right]} \nonumber \displaybreak[0] \\
&\approx \frac{\sum_{i \in \mathcal{I}} \frac{1}{N_i} \sum_{n=1}^{N_i} w_i^n f(\hat{x}_i^n)}
{\sum_{i \in \mathcal{I}} \frac{1}{N_i} \sum_{n=1}^{N_i} w_i^n} \quad \text{where} \quad
\hat{x}_i^n \sim q(x | A_i), \quad w_i^n := \frac{\gamma(\hat{x}_i^n)}{q(\hat{x}_i^n | A_i)}. \label{eq:expt-comb}
\end{align}
We now see that we can calculate estimates separately for each region $A_i$
and then combine these in an unweighted manner -- there are no
correction factors for the strategy used to assign computational resources.
However, we emphasize that there are two key reasons that we are able to do this.  Firstly, rather than
locally self-normalizing, we separately combine unnormalized target estimates and an estimate for the normalization constant, and then \emph{globally self-normalize} the estimate.  Secondly, the truncated 
proposals $q(x|A_i)$ are \emph{correctly normalized} such that $\int_{x \in A_i} q(x|A_i) dx = 1$.

In practice, we often do not know $f(x)$ at inference time. However, we can always compute empirical measures based on weighted samples
$\frac{1}{N_i}\sum_{n=1}^{N_i} w_i^n \delta_{\hat{x}_i^n} (\cdot)$, which can then later be used to evaluate any target function as and when required.

Though the leaves of an inference tree form a suitable disjoint partitioning of the target
space, we also have access to local estimates from non-leaf nodes, left over from when those
nodes were previously leaves themselves.  The IT estimator is therefore constructed
recursively, such that the estimate at any node is a combination of its child estimates
and this local estimate; the propagation step of the algorithm corresponds to
online updates of these estimates.  
To combine estimates from parents with the children, 
we introduce a preference factor to the estimator
from the child nodes, $c_j \in [0,1]$,
and define the IT estimator for node $j$ recursively using
\begin{subequations}
\label{eq:pihat}
\begin{align}
\label{eq:emp-measure-prop}
\hat{\pi}_j \left(\cdot\right) := \frac{\hat{\gamma}_j(\cdot)}{\hat{\omega}_j}, \quad \text{where} \quad
\hat{\gamma}_j(\cdot) &:= \frac{(1-c_j)}{N_j} \sum\nolimits_{n=1}^{N_j} w_j^n \delta_{\hat{x}^n_j} (\cdot)
+c_j \left( \hat{\gamma}_{l_j} (\cdot) +  \hat{\gamma}_{r_j} (\cdot) \right),
\\
\hat{\omega}_j &:= 
\frac{(1-c_j)}{N_j} \sum\nolimits_{n=1}^{N_j} w_j^n
+c_j \left( \hat{\omega}_{l_j}+  \hat{\omega}_{r_j}\right), \label{eq:ml-estimate}
\end{align}
\end{subequations}
$\ell_j$ and $r_j$ refer to the child node indices, and our overall estimate is given
by that of the root node $\hat{\pi}_0\left(\cdot\right)$.  For leaves, $c_j=0$, such that
we simply take the local
estimate.  For internal nodes, let $M_j$ denote the total number of samples drawn at that node
or any of its descendants.  We then define $c_j = \chi_j(M_j-N_j)/M_j $  (such that 
$(M_j-N_j)/M_j $ is the proportion of the samples that are from the children)
and $\chi_j$ is an additional factor to account for the fact that the child estimate will generally
be more efficient than the parent (see Appendix~\ref{sec:app:child-pref}).  Critically, $c_j \to 1$ as
$M_j\to\infty$ for a fixed $N_j$.

The IT approach is backed up by the following consistency result in the number of IT iterations.  
\begin{restatable}{theorem}{it}
	\label{the:est-main}
	If the following hold as the number of IT iterations becomes infinitely large
	\vspace{-2pt}
	\begin{itemize}
		\itemsep-0.05em 
		\item[-] The total number of leaf nodes remains bounded
		and each is visited infinitely often;
		\item[-] When provided with an infinite	sample budget and an
		arbitrary subregion $A$ generated by the node splitting procedure,
		the base inference algorithm produces  an
		empirical measure $\hat{\gamma}(\cdot)$
		 and normalization constant estimate $\hat{\omega}$  
		which respectively converge weakly to 
		$\gamma(x)\iden(x \in A)$ and 
		converge in probability to $\int_{x \in A} \gamma(x) dx$;
	\end{itemize} 
	\vspace{-2pt}
	then each $\hat{\pi}_j (\cdot)$ as defined by~\eqref{eq:pihat}
	converges weakly to $\pi(x|x\in A_j)$ and, in
	particular,
	$\hat{\pi}_0 (\cdot)$ converges weakly to $\pi(x)$.
\end{restatable}
\vspace{-3pt}
The proof is given in Appendix~\ref{sec:app:alltheorm}.
We see that, subject to mild assumptions,
consistency is achieved regardless of our traversal and refinement strategies.  
Inevitably, however, these will affect the practical performance. In the following, we now
develop effective strategies for each component in turn.


\subsection{Partitioning the Target Space}
\label{sec:projection}

Directly partitioning in the space of $x$ can be challenging. Typically it will not be desirable for the partitions to be axis-aligned (or even linear). Conversely, it is in general not possible
to evaluate the partitioned proposal $q(x|A_j)$ for arbitrary $A_j$.
To address this, ITs use a reparameterization of the \emph{proposal}, such that $x=g(z_{1:T})$
where $z_{1:T}$ is uniformly distributed on the unit hypercube $[0,1]^T$.  Though this
is not always exactly the case, $g$ can generally be thought of as an inverse cumulative distribution
function of $q(x)$ (with $T$ set to the dimensionality of $x$).  ITs use axis-aligned
partitions in the space of $z_{1:T}$, which in turn induce (typically nonlinear) partitions on
$x$.  The motivation for this is twofold.  Firstly, because the distribution is uniform over $z_{1:T}$, this
eliminates the problem of trying to choose splits that align well with the contours of $q$.
Secondly, it means that we can easily sample from and evaluate $q(x | A_j)$: $A_j$ will
always represent a hyperrectangle $B_j$ in the space of $z_{1:T}$, so we can sample
from $q(x | A_j)$ by sampling uniformly from $B_j$ and passing the samples through $g$, while 
$q(x | A_j) = q(x) \iden(x \in A_j) / \lVert B_j\rVert$ where $\lVert B_j\rVert$ is the 
volume of this hyperrectangle, leading to simple evaluation as required by the importance
weight evaluations.  See Appendix~\ref{sec:app:projection} for further discussion.

%

%


\section{Traversal Strategy}
\label{sec:traversal}

As explained in \S\ref{sec:overview}, the traversal strategy starts at the root node and
then recursively chooses the child node with the higher utility $u_j$ until a leaf node
is reached.  Though we will use a utility of the UCB form given in~\eqref{eq:ucb-str},
our reward estimate $\hat{r}_j$ will reflect both the need for
exploitation and exploration, unlike in standard approaches where it represents only
exploitation.  

We start by quoting our final choice for the utility, before explaining each of the component terms in
detail.  Using $\text{pa}(j)$ and $\text{si}(j)$ to denote the parent and sister of node $j$ respectively, we have
\begin{align}
\label{eq:uct_it}
u_{j}
=\frac{1}{M_{j}}\Bigg(\left(1-\delta\right)\left(\frac{\hat{\tau}_{j}}{\hat{\tau}_{\text{pa}(j)}}\right)^{\left(1-\alpha\right)}+ \delta \frac{\hat{p}^s_{j}}{\hat{p}^s_{j}+\hat{p}^s_{\text{si}(j)}}
+ \beta \frac{\lVert B_{j} \rVert}{\lVert B_{\text{pa}(j)} \rVert} \frac{\log M_{\text{pa}(j)}}{\sqrt{M_{j}}}\Bigg)
\end{align}
Here $\hat{\tau}_j$ estimates the optimal asymptotic rate for sampling the node (see~\eqref{eq:exploit-reward}), 
while $\hat{p}^s_{j}$ is a subjective probability estimate for the node containing significant 
posterior mass (see~\eqref{eq:den-est}).  Consequently, the first and second terms encourage exploitation 
and targeted exploration respectively, with
$\delta \in [0,1]$ being a parameter that controls the relative emphasis.  We will typically reduce $\delta$
over time to encourage more exploitation, along with $\alpha \in [0,1]$, an annealing
parameter that encourages sampling of the tails.
  The
  different normalizations for $\hat{\tau}_j$ and $\hat{p}^s_{j}$ 
  originate from the fact that we want the exploration term to dominate
  whenever $\hat{\tau}_{\text{pa}(j)} \gg \hat{\tau}_{j}+\hat{\tau}_{\text{si}(j)}$,
  implying that the children have underestimated the exploitation target.
The last term is a classical optimism boost~\cite{auer2002finite}, with the exception that it is scaled by the
relative volume of the node $\lVert B_{j} \rVert / \lVert B_{\text{pa}(j)} \rVert$.
We now discuss  $\hat{\tau}_j$ and $\hat{p}^s_{j}$ in detail.

\subsection{Exploitation Target}

To derive our exploitation target, $\tau_j$, we ask the question: what is the asymptotically optimal rate
for allocating samples to regions?  In other words, if all our node estimates were perfect, how should
we allocate our samples?  One might intuitively expect that the answer to this would be to allocate 
samples in proportion to the marginal probability mass of a region.  However, it turns out that this is 
not the case: the variance on the weights is different for different regions and so we also need to sample more
from regions where this variance is high.  In fact, as we show in Appendix~\ref{sec:app:ess}, the optimal
allocation strategy is to sample according to
\begin{align}
\label{eq:exploit-reward}
\tau_{j} = \sqrt{\omega_{j}^2+(1+\kappa) \sigma_{j}^2} 
\end{align}
where $\sigma_{j}^2$ is the variance of the weights (as produced by single traversal) and
$\omega_{j}$ is the marginal posterior mass of the region as before. Here $\kappa \in [0,\infty]$ is a
``smoothness'' parameter, which dictates the relative importance of the two terms when using the generated samples
to estimate a particular expectation as per~\eqref{eq:expt-comb}.  For example, $\kappa \to \infty$ 
corresponds to the optimal setup for estimating the marginal likelihood, for which $f(x)=1$ is completely flat.

To estimate $\tau_{j}$, we use $\hat{\tau}_{j} = \sqrt{\hat{\omega}_{j}^2+(1+\kappa) \hat{\sigma}_{j}^2}$ 
with propagated estimates $\hat{\omega}_{j}$ and $\hat{\sigma}_{j}^2$.  The former is given by~\eqref{eq:ml-estimate},
while the latter requires a distinct propagation scheme as discussed in Appendix~\ref{sec:app:ess}.

\subsection{Targeted Exploration through Density Estimation of the log Weights}

Relying only on the optimism boost for exploration, as done by standard UCB schemes,
can be chronically inefficient in practice as it only encourages a uniform exploration.
We, therefore, introduce a targeted exploration term into our utility, $\hat{p}_j^s$, which provides a subjective
probability estimate for the event that the region contains significant posterior mass that we have thus far missed.

Providing such a reliable estimate is a challenging problem. Our global proposal $q(x)$ is often very
poor meaning standard MC estimates can be woefully inadequate: we will consider experiments where we regularly
underestimate the marginal likelihood (ML) by factors in excess of $10^{1000}$.


Our insight is that, even when the ML is substantially underestimated, the raw log weights
still convey useful information about what the true value \emph{could} be.
We exploit this insight by carrying out \emph{density estimation of the log weights}  and using this as
a basis for constructing $\hat{p}_j^s$.  Consider the demonstrative
\begin{wrapfigure}{r}{0.45\textwidth}
	\centering
	\includegraphics[width=0.45\textwidth,trim={4cm 0 4cm 0},clip]{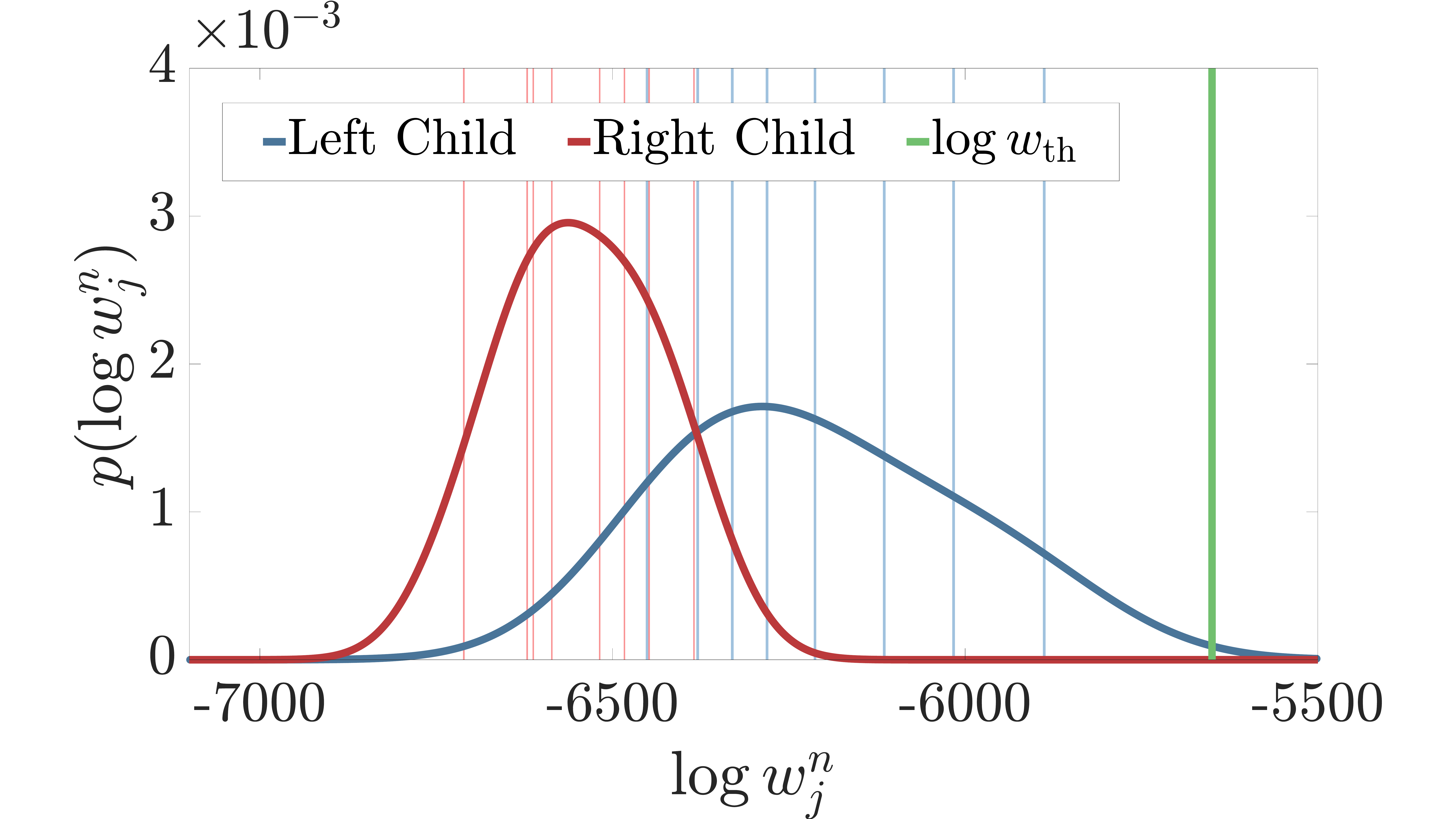}
	\vspace{-18pt}
	\caption{Density estimation for log weights.  
		\label{fig:den-exp}}
			\vspace{-15pt}
\end{wrapfigure}
example shown in Figure~\ref{fig:den-exp} where we want to predict whether the true
log ML of each child is above some threshold $\log w_{\text{th}}$.  Here we 
see that there is a high chance that the left child has a true
ML above the threshold, but we can be reasonably confident the right does not.  
Critically, we can make this assertion even though our MC estimates for the ML
are underestimated by hundreds of orders of magnitude.

To formalize this intuition, let $\psi(\log w_j^n)$ denote a density estimator for a nodes local weights, 
with associated cumulative density $\Psi(\log w_j^n)$.  The key idea 
is to use this density estimator to predict the probability that one more samples will exceed a target
threshold $\log w_{\text{th}}$ if we were to generate another $T$ ``lookahead'' samples, where $T$ is some large, but finite, number.
When $\log w_j^n$ varies over a large range, the MC estimate for the ML is effectively equal to
the maximum weight, and so we have
\begin{align}
P\left(\hat{\omega}_j(T) > w_{\text{th}}\right) \approx
P\left(\max(w_j^{1:T})>w_{\text{th}}\right) \approx
1-(1-\Psi(\log w_{\mathrm{th}}))^T
\end{align}
where $\hat{\omega}_j(T)$ is MC estimate for the ML after taking $T$ samples.
Though we could now use this estimate to construct $\hat{p}^s_j$ directly, 
we apply a heuristic of scaling by the effective sample size (ESS)~\cite{owen2013mc} of the node (see 
Appendix~\ref{sec:app:ess}) on the basis that a high ESS suggests that we have already
a reasonable ML estimate and thus do not need to explore further.  

To complete the picture, we define the propagation strategy for these probability estimates by assuming
that the $\hat{p}_j^s$ are independent for sibling nodes, finally yielding the recursive definition\footnote{In practice, 
	we also use some additional heuristics, giving a slightly different estimator.
	See	Appendix~\ref{sec:app:true-lik}.}
\begin{align}
\label{eq:den-est}
\hat{p}^s_j := \left(1-c_j\right)\frac{1-(1-\Psi(\log w_{\mathrm{th}}))^T}{\text{ESS}_j}
+c_j \left(\hat{p}^s_{\ell_j}+\hat{p}^s_{r_j}-\hat{p}^s_{\ell_j}\hat{p}^s_{r_j}\right)
\end{align} 
analogous to that of $\hat{\omega}_j$ in~\eqref{eq:ml-estimate}.  In our experiments, we found
$\log w_j^n$ was typically well approximated by a Gaussian
(there is also theoretical evidence this is appropriate when SMC
is used as the base algorithm~\citep{berard2014lognormal,doucet2015efficient,pitt2012some})
and so this simple choice was taken for $\psi$.
In cases where this gives a poor fit, one could instead use a kernel density estimator.
Setting $T$ and $w_{\text{th}}$ is detailed in Appendix~\ref{eq:app:int}.

\section{Refinement Strategy}
\label{sec:refinement}

Once a leaf is chosen by the traversal, there are two ways we can refine the tree: 
update the local estimate or split the node.  The two considerations here are whether to split and how to split.  


At a high-level, a good partitioning structure
is one in which the posterior mass is concentrated in a small number of regions.
In essence, we gain most from being able to ``eliminate regions'' from consideration,
reducing the proportion of the target space that needs to be actively considered.
When we propose to split a node, we thus want to find the split that best concentrates
the posterior mass.
Conveniently, we can use the samples already generated at the node to try and
predict what will be a good split.  Namely, we can hypothesize a number of splits and
then evaluate how well each split will concentrate the mass, based on the existing samples.
Though we do not directly use them in this way, ITs indirectly parameterize an importance sampling
proposal, whereby we traverse the tree, recursively sampling a child with probability
proportional to $\hat{\tau}_j$.  We can, therefore, measure the concentration of mass through the entropy
of this implied proposal.

Recall from \S\ref{sec:projection} that ITs use axis-aligned partitions in the reparameterized 
space $z_{1:T}$ and that our proposal for a leaf node is uniform in this space.  We can therefore
analytically calculate the entropy of a hypothetical split (see Appendix~\ref{sec:app:refine})
and use this as loss criterion for choosing a split:
\begin{align}
\label{eq:split-criterion}
\textsc{Loss}(\text{split}) & 
= 
\hat{\omega}_{\ell} \log\frac{\lVert B_{\ell}\rVert}{\hat{\omega}_{\ell}} + 
\hat{\omega}_{r} \log\frac{\lVert B_{r}\rVert}{\hat{\omega}_{r}}
\end{align}
where the child volumes and marginal probability estimates are implied for any
hypothetical split. The lower this loss, the more
information our split conveys about where the posterior mass is concentrated.  As
hypothetical splits can be quickly tested -- there
is no need to run inference -- we can efficiently test out
a relatively large number ($\sim100$)  of random splits and then choose the one that minimizes~\eqref{eq:split-criterion}.
We then initialize the newly generated nodes by running inference separately on
each of them.
%
%

We further introduce heuristics for whether to split in order to avoid unnecessary over-splitting.  
Firstly, we only attempt to split once $N_j$ reaches a certain threshold
and if the ratio $\text{ESS}_j/N_j$ falls below a certain threshold: we want to stop splitting once
a node represents a near perfect sampler.  Secondly, whenever we split a node, we
check that split passes a usefulness test, namely a significance test that the distributions of the 
$\log w_j$ are different, rejecting the split if this test fails.

\section{Experiments}
\label{sec:experiments}

\begin{figure*}[t!]
	\centering
	\includegraphics[width=0.4\textwidth]{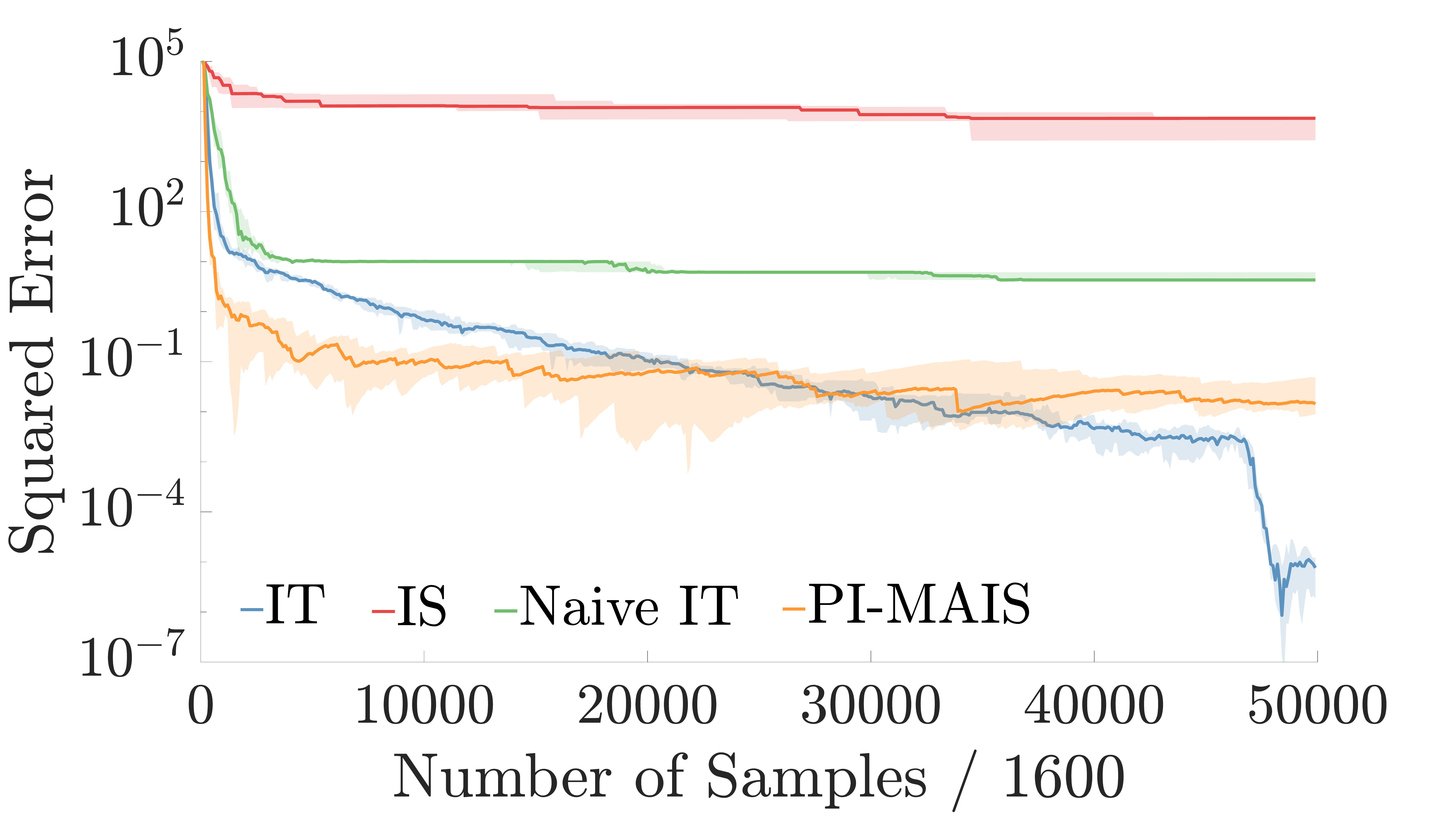} \hspace{20pt}
	\includegraphics[width=0.4\textwidth]{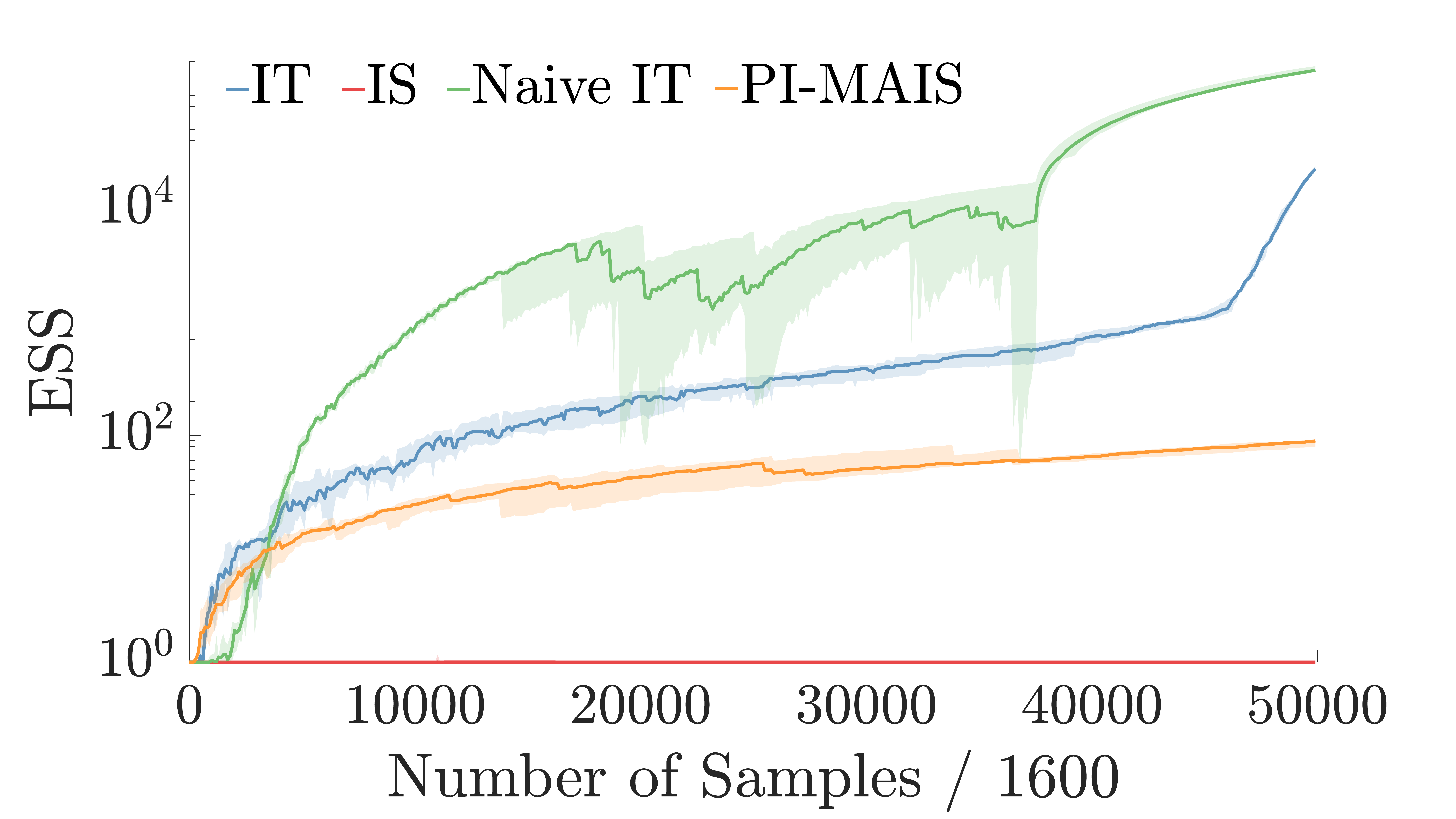}
	\vspace{-5pt}
	\caption{
		Convergence for the GMM in terms of the log ML estimate (left) and
		the ESS (right).  The ground truth log marginal was estimated using
		a very large number of samples and a manually adapted proposal.  Solid lines represent
		median over 10 runs and shading the 25\%-75\% quantiles. The reason for the ``just-in-time''
		style convergence of the IT stems from the fact that the parameter annealing 
		schedules start to kick in and encourage far more exploitation near the end of the runs.
		\label{fig:gmm-conv}}
\end{figure*}

\subsection{Gaussian Mixture Model}
\label{sec:gmm}

Our first experiment is to infer the cluster means in a Gaussian mixture model (GMM).
Specifically,
\begin{align*}
{\mu}_k \sim \mathcal{N}({0}, \Sigma_{\mu}), \quad \,\,
z_n \big|\ {\pi} \sim \mathrm{Categorical}(\{1/K,\dots,1/K\}), \quad \,\,
{y}_n \big|\ z_n = k , {\mu}_k \sim \mathcal{N}({\mu}_{k}, \Sigma_{y}),
\end{align*}
where we set $\Sigma_{\mu} = {I}$,  $\Sigma_{y} = 0.2{I}$,
and $K=4$.  We generated a two-dimensional synthetic dataset ${y}_{1:200}$ using the generative model
and then ran ITs with importance sampling as the base algorithm to conduct inference on ${\mu}_k$,
with the $z_n$ marginalized out by summation.  We use the prior on ${\mu}_k$ as our base proposal.
Though simple, this constitutes a surprisingly
challenging inference problem, as symmetries in the model
mean that the posterior is concentrated in $24$ well-separated modes, each of which occupy less than $10^{-10}$ of the 
overall eight-dimensional parameter space.

For computational efficiency, we fixed ``one run'' of the base inference algorithm to be
comprised of drawing $100$ importance samples and we undertook $16$ runs of
this base algorithm for each refinement step (with each counting as a separate traversal). 
We further took the convention in, for example, log weight density estimation that
each ``run'' returns a single amalgamated $w_i^n$, which might itself contain multiple samples (similarly
$w_i^n$ becomes the SMC ML estimate in the next experiment).
We compared to the following baselines given the same total budget of target density evaluations:
non-adaptive importance sampling; a na\"{i}ve IT implementation where we set $\delta=0$, $\alpha=0$, and $\beta=0.5$,
which means that our target ignores the $\hat{p}_j^s$ terms and relies solely on the optimism boost for
exploration; and PI-MAIS~\citep{martino2017layered}, a state-of-the-art adaptive importance sampler based
on simulating a large number of Markov chains to construct the proposal.  
Each algorithm was given a budget of $8\times10^7$ target evaluations, with the parameters set 
as per Appendix~\ref{sec:app:exp}.

For comparison, we examined the convergence of the ML estimate
and ESS (Figure~\ref{fig:gmm-conv}) and a kernel density estimator of the final 
output (Figure~\ref{fig:gmm-den}).  The results show that ITs outperformed
the alternatives.  Unsurprisingly, vanilla importance sampling performed poorly throughout,
ending with an ESS of effectively 1.  The na\"{i}ve IT implementation managed to generate
a very high ESS, but typically only found two or three modes leading to a substantial error
in the ML estimate.  PI-MAIS did better at finding modes, though still substantially worse than IT.
Further, it ended with a low ESS and produced poor estimates for the relative sizes
of the modes, in turn giving an inferior log ML estimate.

\subsection{Chaotic Dynamics Model}
\label{sec:chaos}

Dealing with long-range dependencies, i.e.~variables that have influence many steps 
after they are sampled, can be challenging in SMC as variables are often fixed before all dependent 
terms are incorporated, leading to sample degeneracy.  Viewing this in another light, 
the intermediate target distributions can vary substantially from the target marginal
distribution on the relevant variables.
Na\"{i}ve strategies for dealing with this tend to be futile -- the resampling step 
always corrects to the intermediate target and thus incorporating lookahead information 
in proposals 
often reduces the effective sample size.  
In some cases, auxiliary weighting schemes provide a degree of
lookahead~\citep{doucet2006efficient,lin2013lookahead}, but 
these typically entail a substantial
increase in computational cost while providing only a short-range lookahead. Moreover, problems with degeneracy
can be compounded in the context of adaptation as information is only received for particles that survive the
resampling.  We now show that ITs can address 
these challenges by running inference on separate regions.  Namely, the IT process allows information to be gathered even in the face of
degeneracy. Constraining different sweeps to different regions allows samples to be 
``forced through'' the resampling steps, hereby dealing with long-range dependencies.  This is done without
losing the key benefits of SMC, as gains from resampling are still seen when running inference
within a particular region.  Note that ITs only require an unbiased estimate for the weights
in a manner akin to pseudo-marginal methods~\cite{andrieu2009pseudo},
such that
we can run SMC when there are some latent variables not directly controlled by the IT.

%
%

\begin{figure}[t!]
	\centering
	\subfigure[PI-MAIS]{\label{fig:tree-smc-IT}\includegraphics[height=0.155\textwidth]{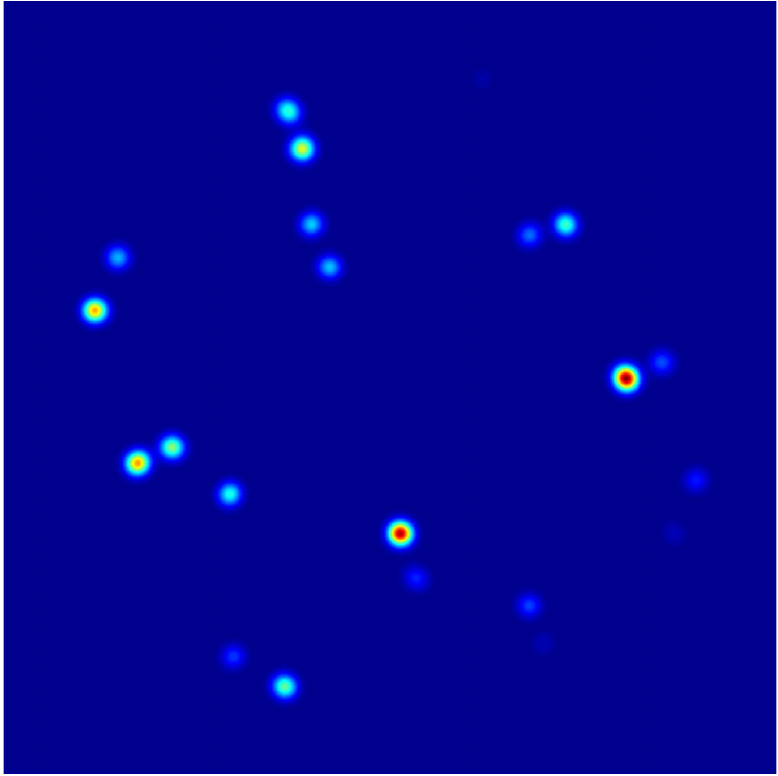}} \hspace{-2pt}
	\subfigure[Na\"{i}ve IT]{\label{fig:gmm}\includegraphics[height=0.155\textwidth]{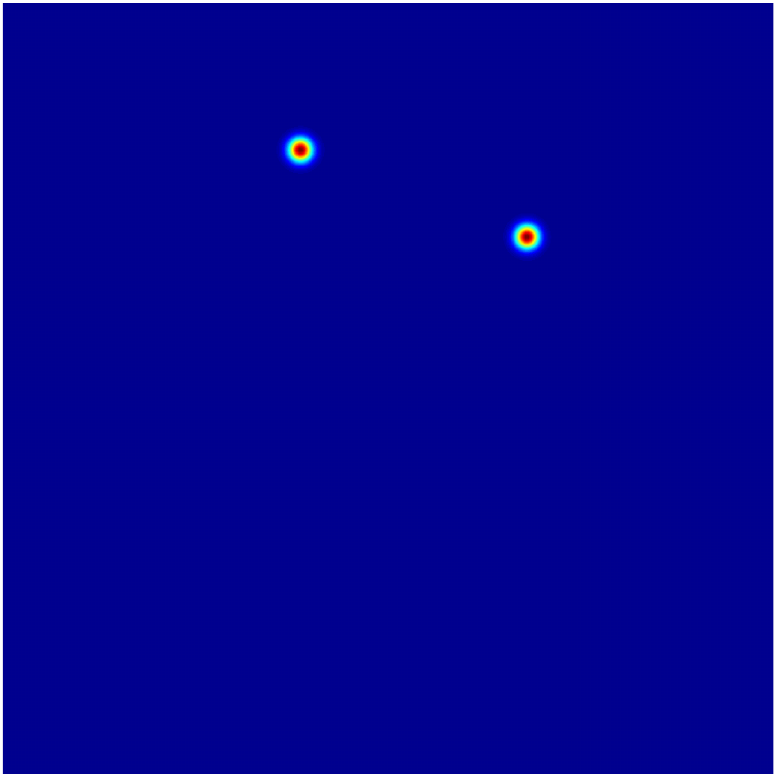}} \hspace{-2pt}
	\subfigure[IT]{\label{fig:network}\includegraphics[height=0.155\textwidth]{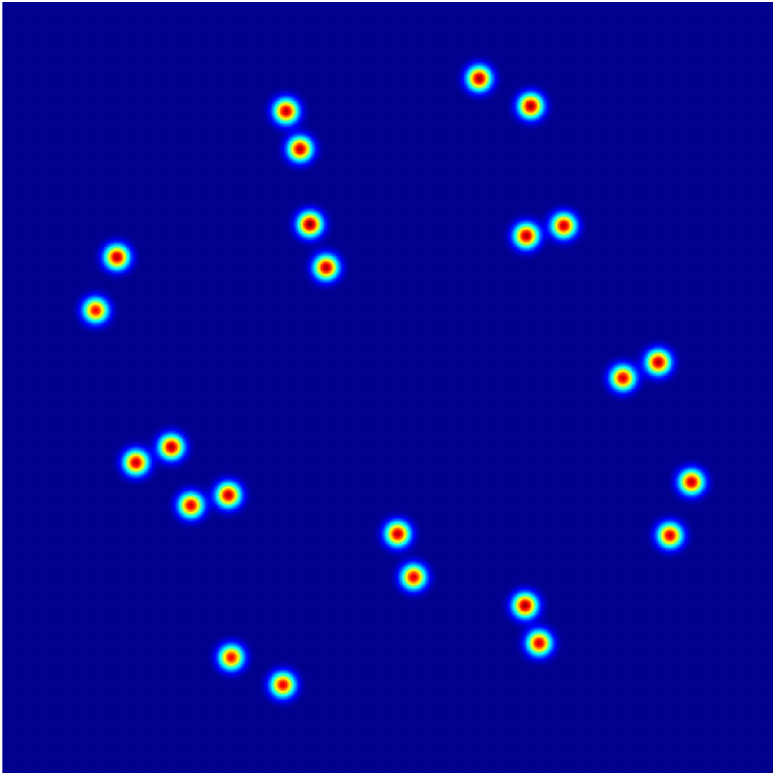}} \hspace{-2pt}
	\subfigure[PMMH]{\includegraphics[height=0.155\textwidth]{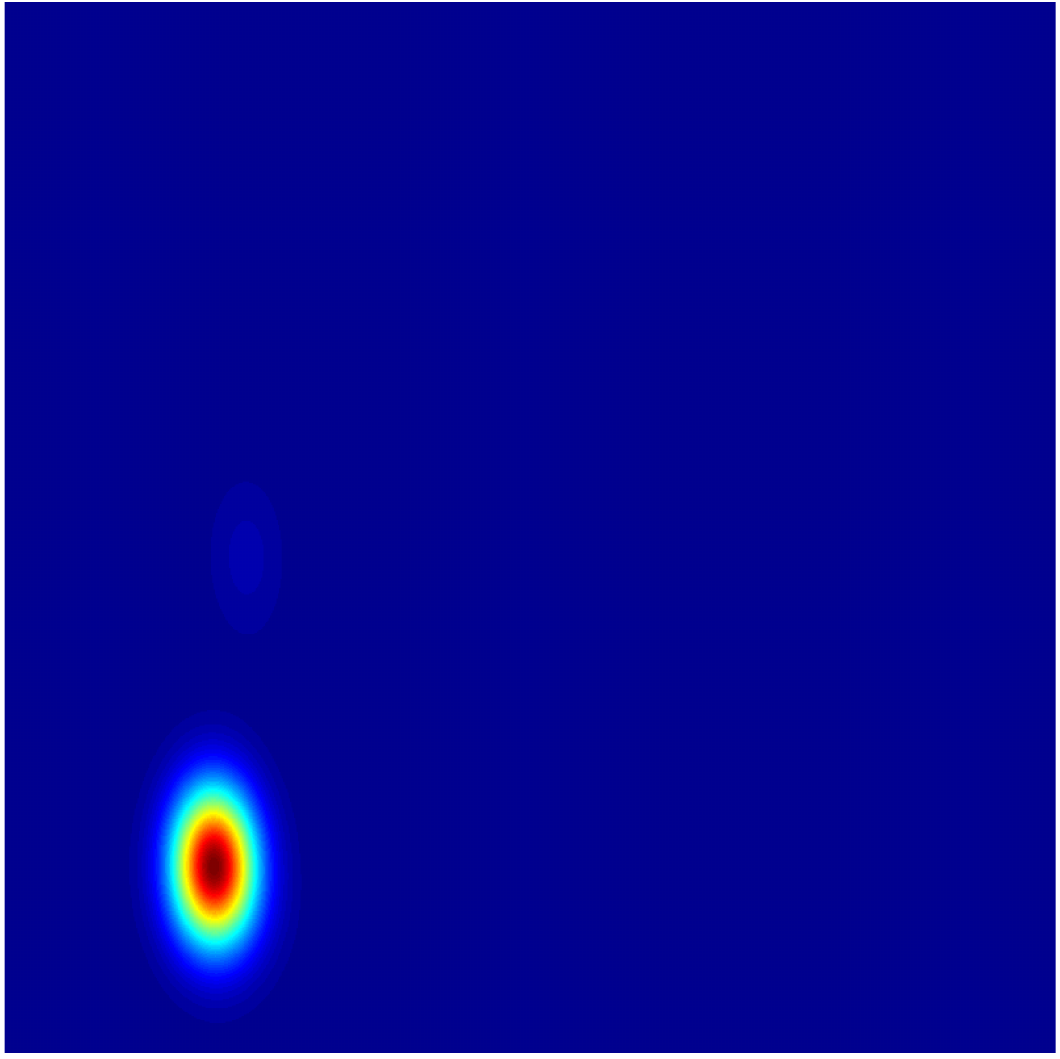}} \hspace{-2pt}
	\subfigure[Na\"{i}ve IT]{\includegraphics[height=0.155\textwidth]{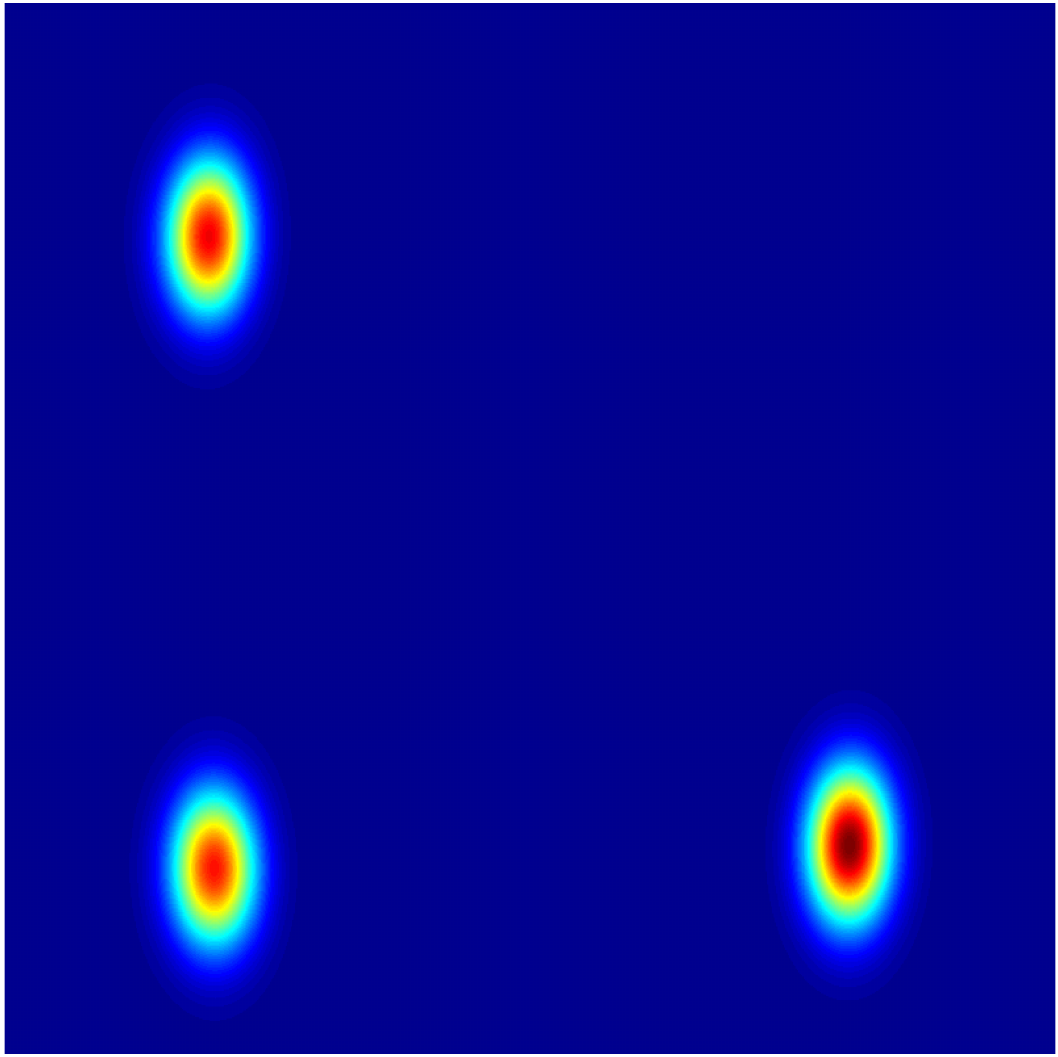}} \hspace{-2pt}
	\subfigure[IT]{\includegraphics[height=0.155\textwidth]{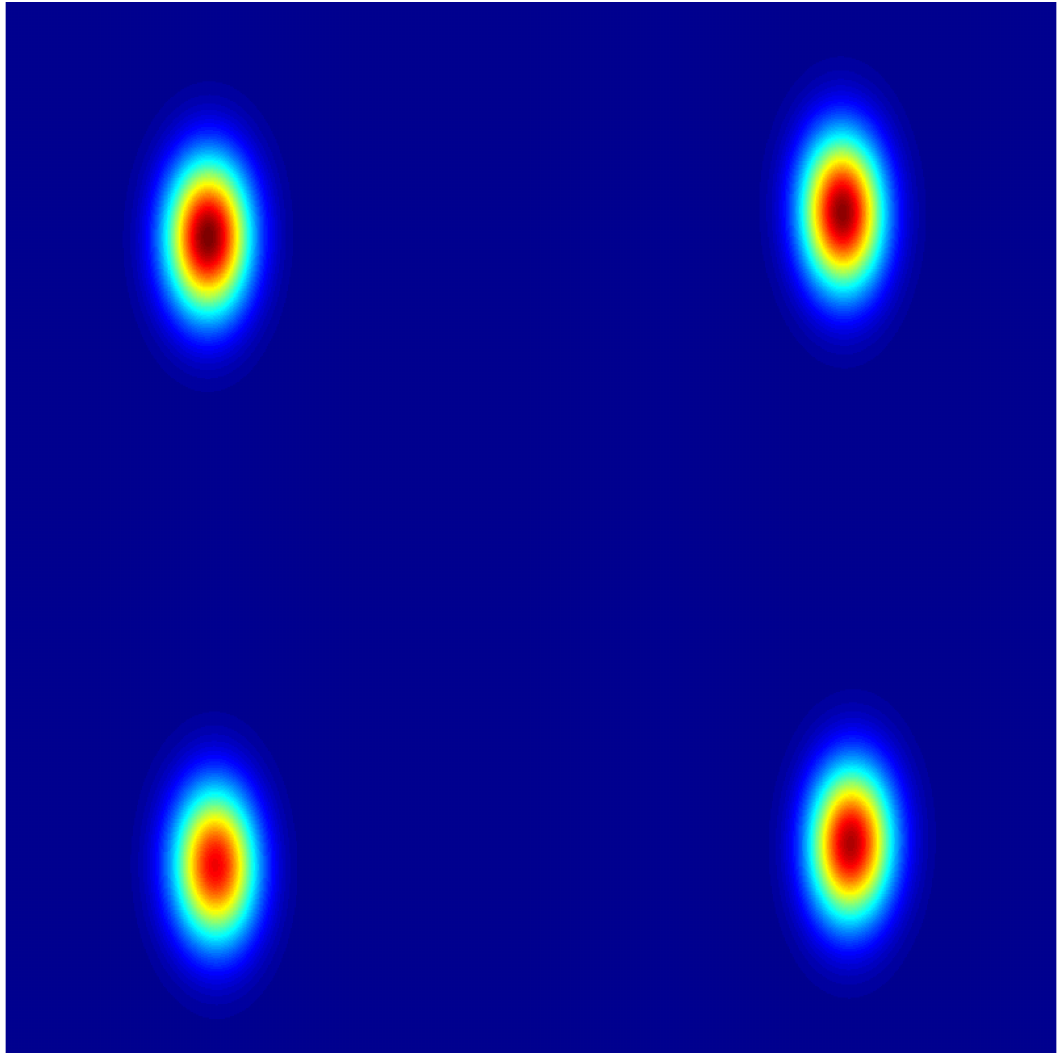}}%
	\vspace{-4pt}
	\caption{Kernel density estimate of projected posterior estimates for the GMM (a-c) and the chaos model (d-f).  
		We use a linear projection
		of the original 8/4-dimensional spaces and exaggerate the
		variance of the modes for visualization purposes. For both problems, the IT has successfully recovered
		all modes and inferred that all the modes have equal mass. Though the 
		na\"{i}ve IT implementation produced good estimates for the modes it found, it missed modes
		for both problems. For the GMM, PI-MAIS found a number of modes but still missed some and misestimated their relative masses.  For the chaos model, PMMH only found a single mode. 
		\vspace{-12pt}
		\label{fig:gmm-den}}
\end{figure}

\begin{figure*}[t!]
	\centering
	\includegraphics[width=0.4\textwidth]{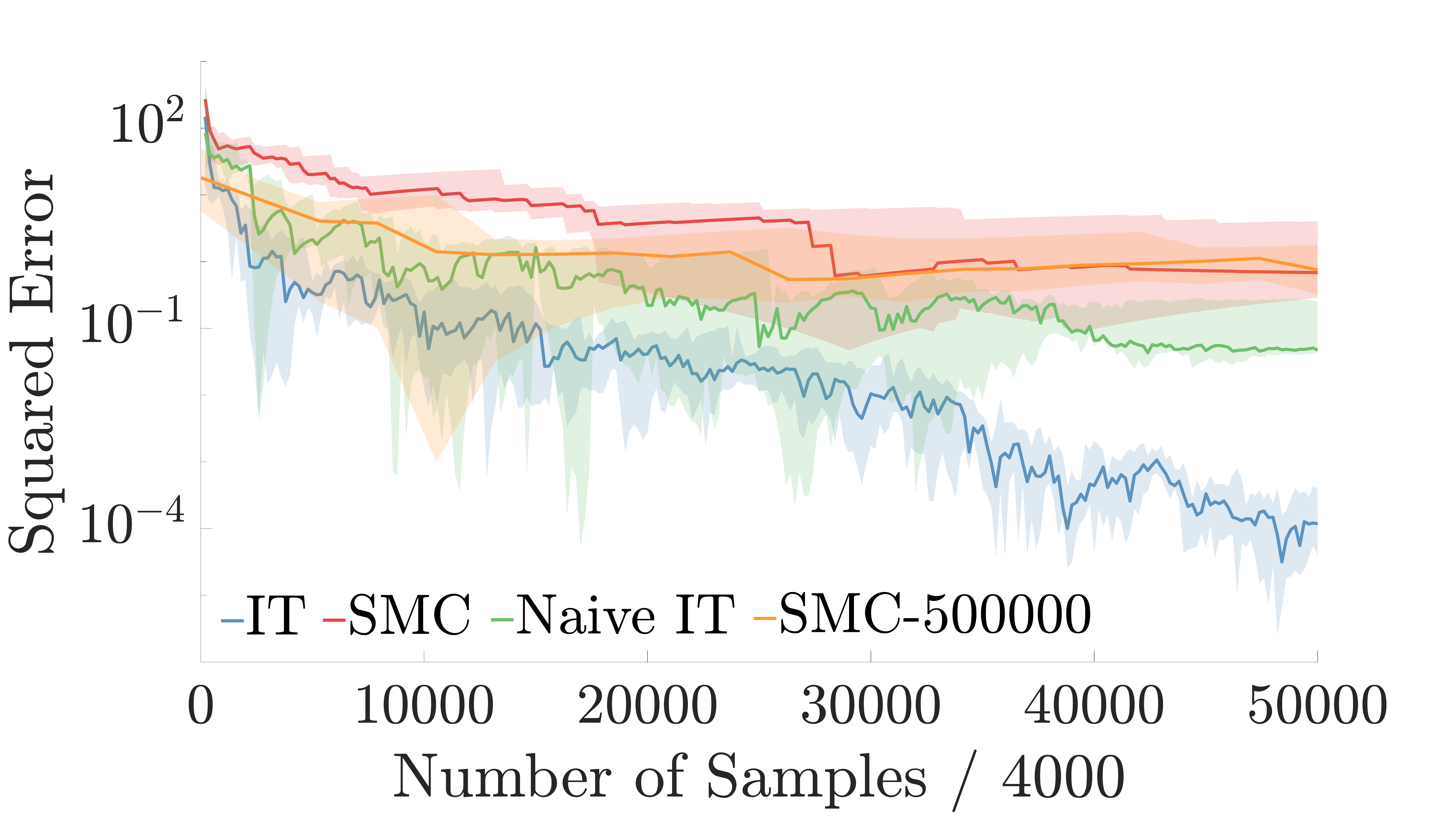} \hspace{20pt}
	\includegraphics[width=0.4\textwidth]{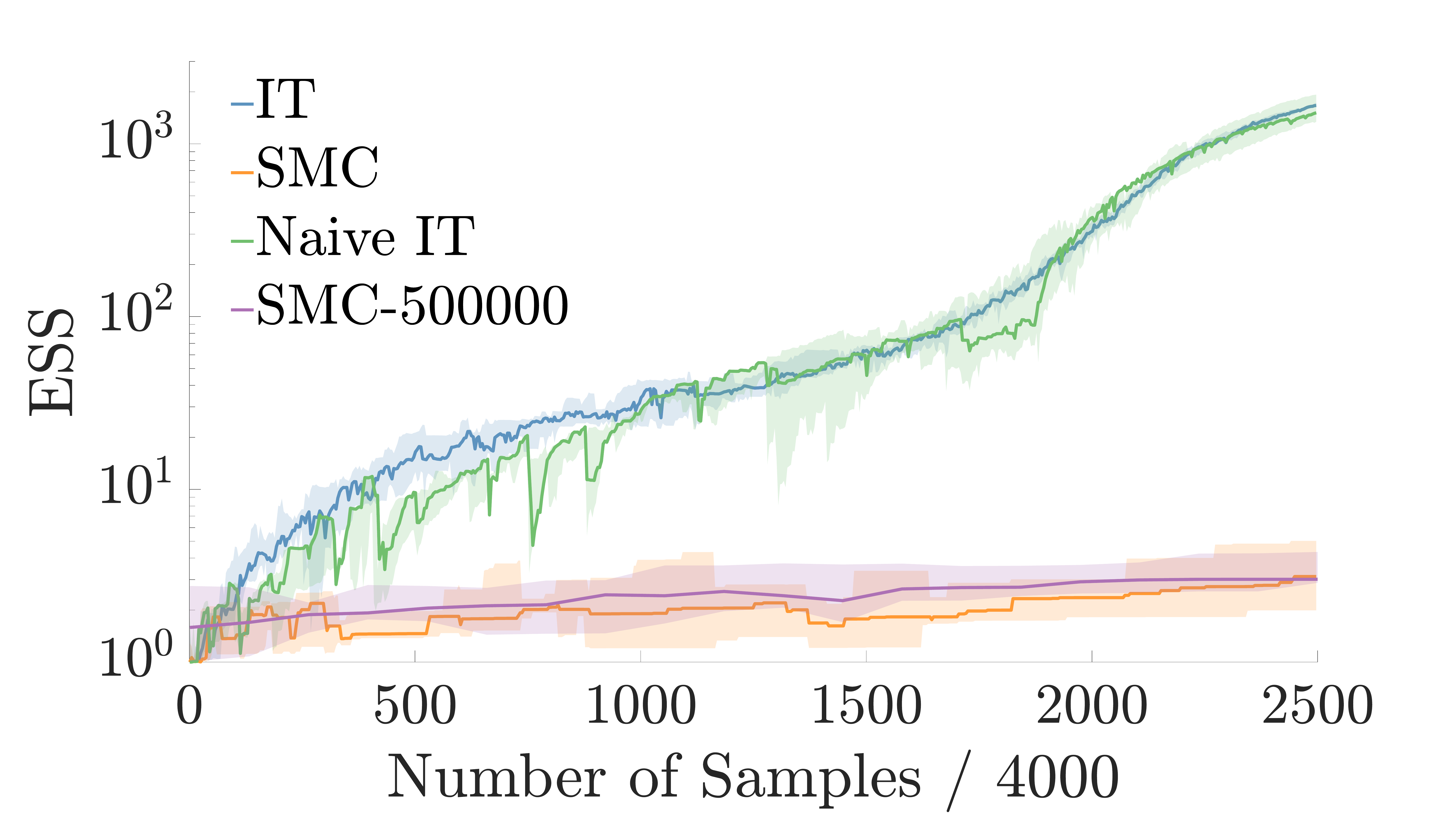}
	\vspace{-4pt}
	\caption{Convergence of log ML and ESS for chaos model,
		conventions as per Figure~\ref{fig:gmm-conv}.  PMMH is not shown as it returns unweighted samples
		and no ML estimate; other results are given in Appendix~\ref{sec:app:pmmh}.  
		\vspace{-7pt}
		\label{fig:chaos_conv}}
\end{figure*}

To test ITs in this setting, we consider an adaptation of the chaotic dynamical system tracking problem 
introduced by~\cite{rainforth2016BOPP}, details for which are given in Appendix~\ref{sec:app:exp}.  
The model comprises of an extended Kalman filter where we have dynamics parameters $\theta$, latents
$x_{1:T}$, and observations $y_{1:T}$.  We desire to conduct inference over both 
the dynamics parameters and the latent variables, but
will only use ITs to control the sampling of the former.  This model
contains long-range dependencies because the dynamics parameters affect each
transition and so the smoothing marginal $p(\theta | y_{1:T})$ is very different to 
the filtering marginal $p(\theta | y_{1})$.  In fact, the two are so different that
using the so-called one-step-optimal proposal, the target for most methods
of SMC proposal adaptation~\cite{gu2015neural}, provides no noticeable performance
improvement over simply sampling from the prior.

Because PI-MAIS requires an MCMC sampler to be run on the target $p(\theta | y_{1:T})$,
it is inappropriate for this problem.  We instead compare to using SMC
without adaptation, SMC with 1000 times more particles, the na\"{i}ve IT implementation,
and PMMH~\cite{andrieu2010particle}, a method explicitly designed for dealing
with global parameters in SMC.  We allowed a budget of $1\times10^7$ target evaluations and used
$8$ SMC sweeps of $500$ particles per refinement step for the IT approaches.
Details on parameters setups are given in Appendix~\ref{sec:app:exp}.
We used the same comparison metrics as for the GMM, with results shown in Figures~\ref{fig:gmm-den} 
and~\ref{fig:chaos_conv}.
We see that ITs again outperformed the other methods.

\section{Conclusions}	
\label{sec:discussion}	

We have introduced inference trees (ITs), a new adaptive inference algorithm
drawing on ideas from Monte Carlo tree search.  We have shown that, by
carrying out explicit \emph{exploration} in the adaptation process, ITs can
avoid common pathologies with other adaptive schemes and reliably
uncover multiple modes.  We have consequently found that, for
the tested models, ITs outperformed previous state-of-the-art adaptive importance
sampling and particle MCMC methods.  In addition to the immediate
utility of the proposed approach, we believe that the general IT framework
opens up many opportunities for new research, due to the separation between
their consistency and the specifics of the learning algorithm.
For example, ITs can also be used for integration (see Appendix~\ref{sec:app:det}).

\clearpage
\appendix
%
	
\section{Additional Details on Partitioning the Target Space}
\label{sec:app:projection}

As explained in the main paper, effectively partitioning in the space of $x$ is difficult and so
we perform a reparameterization of the \emph{proposal} to a  
``cumulative distribution space'', 
such that $x = g(z_{1:T})$ and each $z_t \sim \textsc{Uniform}(0,1)$.  In this
reparameterized space, we use axis aligned partitions, such that any region
can be defined using
\begin{align}
B_j := \zeta_1^j \times \zeta_2^j \times \dots \times \zeta_T^j
\end{align}
where each $\zeta_t^j \subseteq [0,1]$ is a partition for the corresponding dimension of
$z_t$.  These partitions then in turn define partitions on $x$, namely we have
\begin{align}
A_j := \{ g(z_{1:T}) : z_{1:T} \in B_j\}
= \{ g(z_{1:T}) : z_1 \in \zeta_1^j  \cap z_2 \in \zeta_2^j \cap \dots \cap z_T \in \zeta_T^j  \}
\end{align}
A high level description of this process is shown below.

\begin{figure*}[h]
	\centering
	\includegraphics[width=0.7\textwidth]{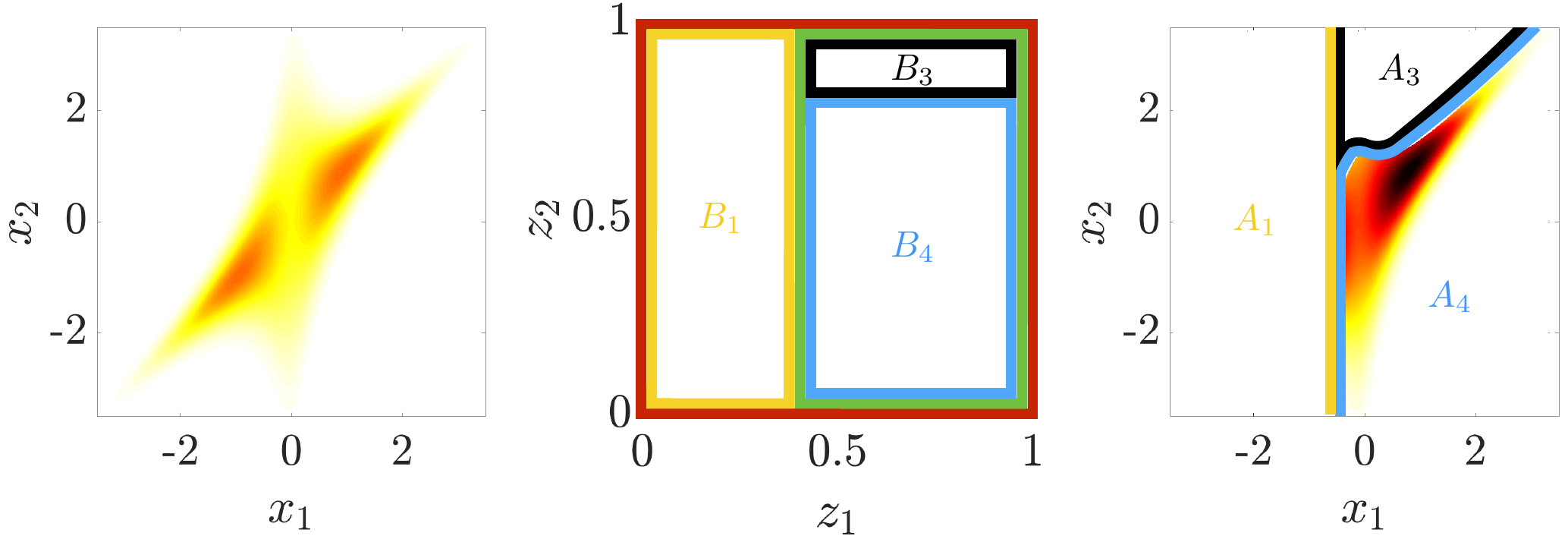}%
	\caption{Truncation of a proposal $q(x_{1:2})$.
		Numbering left to right, [1] shows the original proposal and [2] the hierarchical partitioning
		of $z_{1:2}$ imposed by the tree.
		[3] shows the partitioning implied by $q(x_{1:2})$ and
		the leaf nodes
		on the target space $x_{1:2}$, where we note that the partition between $A_3$ and
		$A_4$ is nonlinear.  It further shows
		the proposal truncated to $A_4$ and renormalized.
		\label{fig:expl}}
\end{figure*}

In general, $g$ can be thought of as an inverse cumulative distribution function.
Namely, if we presume that $x$ is also $T$ dimensional and our proposal factorizes as
\[
q(x_{1:T}) = q(x_1) q(x_2 | x_1) \dots q(x_T | x_{1:T-1})
\]
then $z_t$ is defined by the series of cumulative distribution mappings
\begin{align}
\label{eq:cum}
z_t := \eta_t(x_{t} ; x_{1:t-1}) = \int_{-\infty}^{x_t} q_t(x_t' | x_{1:t-1}) dx_t',
\end{align}
which in turn implicitly defines $g$.  
As we are free to choose the form of the
proposal, we can always ensure that $g$ can be calculated.  In some scenarios, it might
even be helpful to define $q(x)$ implicitly through $g$.
Note that~\eqref{eq:cum} further implies that the marginal proposals
can be expressed in the form
\[
q_t(x_{1:t}) = g_t(z_{1:t}).
\]
such that we can can sequentially generate $x_{1:T}$, as required in the SMC setting.

Another important point of interest is that it is perfectly permissible for $g$ to map multiple different $z_{1:T}$ to the
same $x$.  For example, this is necessary when $x$ is discrete.  In this scenario, 
the $A_j$ may no longer be disjoint,\footnote{From a practical perspective, we postulate that
it may sometimes be preferable to not perform the reparameterization for discrete variables
and instead directly split these in the space of $x$.} but here we can instead rely on
the law of the unconscious statistician: we can think in terms of performing inference on
$z_{1:T}$ (for which the partitions are disjoint) and then taking the pushforward 
distribution this induces on $x$.  Note that this does not require any algorithmic
changes.

Because the distribution over $z_{1:T}$ is a uniform hypercube, the probability of 
generating an $x$ whose pre-image is in $B_j$ is just the hypervolume of $B_j$ (which is
in turn given by the product of the lengths of $\zeta^j_t$).  Therefore, after drawing
from the truncated proposal $q(x | x \in A_j)$ by sampling $\hat{z}_{1:T,j}^n \sim \textsc{Uniform}(B_j)$ and setting $\hat{x}_j^n = g(\hat{z}_{1:T,j}^n)$, we can evaluate
the corresponding weights using
\begin{align}
w_j^n := \frac{\pi(\hat{x}_j^n)}{q(\hat{x}_j^n | \hat{x}_j^n \in A_j)}
= \frac{\pi(\hat{x}_j^n)}{q(\hat{x}_j^n | \hat{z}_{1:T,j}^n \in B_j)}
= \frac{\pi(\hat{x}_j^n)}{q(\hat{x}_j^n)}\lVert B_j\rVert
\label{eq:prop}
\end{align}
where $\lVert B_j\rVert$ is just the (known) area of $B_j$.

We finish our discussion of partitions the target space by noting that it should
be possible to also adapt proposals within individual regions, in addition to the
adaptation already provided by inference trees.  This can be done by sampling $z_{1:T} | B_j$
from a non-uniform distribution which is learned adaptively,
and adjusting~\eqref{eq:prop} accordingly.

\section{Theoretical Justification}
\label{sec:app:alltheorm}

In this section, we demonstrate the correctness of the IT algorithm.

We first demonstrate that for any partitioning $\{A_i\}_{i \in \mathcal{I}}$ and set of consistent estimators
for each partition, then the combination strategies given in \S\ref{sec:ITest} similarly lead
to consistent estimators.  Moreover, we demonstrate that this convergence holds when we combine multiple
sets of estimators, each with their own partitioning, for example the parent estimator and children estimator
in~\eqref{eq:pihat}.  
At a high-level we make three assumptions: each constituent estimator is consistent in isolation,
each set of estimators only has finite combination weight in the limit of large overall computational budget
if each of constituent region estimators receives a finite proportion of that overall computational budget, 
and the number of each regions is finite for each estimator set.  For exposition, we will, for now, assume that the $A_i$ are
disjoint (in Assumption~\ref{ass:cons}), but we show in Appendix \ref{sec:discrete} how that this assumption can be relaxed to any proposal constructed
from the form given in Appendix \ref{sec:app:projection}.

\begin{assumption}
	\label{ass:cons}		
	Let $\X$ denote the support of $x$.
For every independent estimator set $\ell \in \{1,\dots,L\}$, we are given a) a disjoint
partitioning $\{A_{\ell,i}\}_{i \in \mathcal{I}_\ell}$ of the $\X$ such that $A_{\ell,i} \cap A_{\ell,j} = \emptyset$
for $i\neq j$ and $\bigcup_{i \in \mathcal{I}_\ell} A_{\ell,i} = \X$,
and b) a family $\{\hat{\gamma}^{N_{\ell,i}}_{\ell,i}\}_{i \in \mathcal{I}_\ell}$ of estimated
measures on $\X$
\[
\hat{\gamma}^{N_{\ell,i}}_{\ell,i} \left(\cdot\right) := \frac{1}{N_{\ell,i}} \sum_{n=1}^{N_{\ell,i}}
w^n_{\ell,i} \, \delta_{\hat{x}_{\ell,i}^n} \left(\cdot\right)
\]
for some random variables $w^n_{\ell,i}$ and $\hat{x}_{\ell,i}^n$ such that each
$\hat{\gamma}^{N_{\ell,i}}_{\ell,i}(\cdot)$ converges weakly to the following measure on $\X$ as $N_{\ell,i} \rightarrow \infty$
\[
\gamma(x) \iden(x\in A_{\ell,i}).
\]
Further each marginal probability estimate converges in probability as follows
\[
\hat{\omega}^{N_{\ell,i}}_{\ell,i} := \frac{1}{N_{\ell,i}} \sum_{n=1}^{N_{\ell,i}}
w^n_{\ell,i} 
\pto 
\int_{\X} \iden(x \in A_{\ell,i}) \gamma(dx).
\]
\end{assumption}

\begin{assumption}
\label{ass:c}
Let $k_{\ell} : \real_{\ge 0} \to \real_{\ge 0}$ be combination weight functions which produce unnormalized
combination weights $k_{\ell} \left(N_{\ell}\right)$ when 
provided with the total number of samples used for the corresponding estimator set 
$N_{\ell} = \sum_{i \in \mathcal{I}_\ell} N_{\ell,i}$ such that $\lim_{N_{\ell}\to\infty} k_{\ell}(N_{\ell}) = \infty$
for each $\ell$, each $k_{\ell}(N_{\ell})$ is finite for any finite $N_{\ell}$,
and $\sum_{\ell=1}^L k_{\ell}(N_{\ell}) >0$ whenever $R = \sum_{\ell=1}^L N_{\ell} > 0$.
We further assume that for each estimator set $\ell$, either all of the $N_{\ell,i}$ tend to infinity or none of them.
More precisely, we assume there is a
non-empty subset $\mathcal{L}_0 \subseteq \{1,\ldots,L\}$ such that 
for all $\ell \in \mathcal{L}_0$ and $i \in \mathcal{I}_\ell$,
\[
\lim\limits_{R\to\infty} N_{\ell,i} = \infty
\]
and for all $\ell \not\in \mathcal{L}_0$,
\[
\lim\limits_{R\to\infty} N_{\ell} < \infty
\]
almost surely.
\end{assumption}

\begin{assumption}
	\label{ass:stop}
	Each $\mathcal{I}_{\ell}$ is a finite set.
\end{assumption}

The last of these assumptions can probably be relaxed to $\mathcal{I}_{\ell}$ being a countable set, but as
it will be algorithmically beneficial to ensure that the depth of the tree remains bounded, this case
is of little interest anyway.  The need for the second assumption is to ensure that any individual estimator which
only has finite computational budget in the limit of large overall budget is given zero weight after normalization.

We are now ready to demonstrate the consistency of our estimator combination.
\begin{lemma}\label{the:as} 
        If Assumptions~\ref{ass:cons},~\ref{ass:c} and~\ref{ass:stop} hold, then
        \begin{align} 
                \label{eq:estimator} 
                \hat{\gamma}^{\{N_{\ell,i}\}_{\ell,i}} :=  \frac{1}{\sum_{\ell=1}^{L} k_{\ell} \left(N_{\ell}\right)} 
                \sum_{\ell=1}^{L} k_{\ell}\left(N_{\ell}\right) \sum_{i \in \mathcal{I}_{\ell}} \hat{\gamma}^{N_{\ell,i}}_{\ell,i} 
        \end{align}
        converges weakly to the measure $\gamma(x)$ on $\X$ as $R\to\infty$.
\end{lemma}
\begin{proof}
By assumption we have that 
each $\hat{\gamma}^{N_{\ell,i}}_{\ell,i} $ converges weakly
to the measure $\gamma(x)\iden(x \in A_{\ell,i})$ as $N_{\ell,i}$ tends to $\infty$. Thus, for each $\ell \in \mathcal{L}_0$,
$\sum_{i \in \mathcal{I}_\ell} \hat{\gamma}_{\ell,i}^{N_{\ell,i}}$
converges weakly to the measure
\[
        \sum_{i \in \mathcal{I}_\ell}
        \gamma(x)\iden(x \in A_{\ell,i})
        =
        \gamma(x)\sum_{i \in \mathcal{I}_\ell}
        \iden(x \in A_{\ell,i})
        =
        \gamma(x)
\]
as $R\to\infty$. The estimates for $\ell \notin \mathcal{L}_0$ need not converge
but do not affect the final estimate as 
\[
\lim\limits_{R\to\infty} \frac{k_{\ell} (N_{\ell})}{\sum_{\ell=1}^{L}
	k_{\ell} (N_{\ell})} =0 \quad \forall \ell \notin \mathcal{L}_0.
\]
To show the claim of this
theorem, we now consider an arbitrary bounded continuous function $f : \X \to \real$
for which we have
\begin{align*}
        \int f(x)\, \hat{\gamma}^{\{N_{\ell,i}\}_{\ell,i}}(d x)
        &= 
        \int
        f(x)
        \frac{1}{\sum_{\ell=1}^{L} k_{\ell}\left(N_{\ell}\right)} 
        \sum_{\ell=1}^{L} k_{\ell}\left(N_{\ell}\right) \sum_{i \in \mathcal{I}_{\ell}} \hat{\gamma}^{N_{\ell,i}}_{\ell,i}(d x) 
        \\
        &
        =
        \frac{1}{\sum_{\ell=1}^{L} k_{\ell}\left(N_{\ell}\right)} 
        \sum_{\ell=1}^{L} k_{\ell}\left(N_{\ell}\right) 
        \sum_{i \in \mathcal{I}_\ell} 
        \int f(x) \hat{\gamma}^{N_{\ell,i}}_{\ell,i}(d x) 
        \end{align*}
which using Assumptions~\ref{ass:cons}
        	and~\ref{ass:c} converges as $R\to\infty$ to
\begin{align*}	   
	 \frac{ 
	   \sum_{\ell \in \mathcal{L}_0} k_{\ell}\left(N_{\ell}\right) 
	   \sum_{i \in \mathcal{I}_\ell} 
        \int f(x) \iden(x \in A_{\ell,i}) \, \gamma(d x)}
       {\sum_{\ell \in \mathcal{L}_0} k_{\ell}\left(N_{\ell}\right)}
        &=\frac{1}{\sum_{\ell \in \mathcal{L}_0} k_{\ell}\left(N_{\ell}\right)} 
        \sum_{\ell \in \mathcal{L}_0} k_{\ell}\left(N_{\ell}\right) 
        \int f(x)  \, \gamma(d x)\\
        & 
        =
        \int f(x)  \, \gamma(dx)
\end{align*}
and thus the expectation taken with respect to $\hat{\gamma}^{\{N_{\ell,i}\}_{\ell,i}}$ converges
to the true expectation $\int f(x)  \, \gamma(dx)$.
Now as this holds for an arbitrary $f$, this implies weak convergence as required.
\end{proof}

\begin{corollary}\label{the:post} 
Let
\begin{align}		
\hat{\omega}^{\{N_{\ell,i}\}_{\ell,i}} :=
\frac{1}{\sum_{\ell=1}^{L} k_{\ell} \left(N_{\ell}\right)} 
\sum_{\ell=1}^{L} k_{\ell}\left(N_{\ell}\right) \sum_{i \in \mathcal{I}_{\ell}} \hat{\omega}^{N_{\ell,i}}_{\ell,i}.
\end{align}		
If the assumptions of Lemma~\ref{the:as} hold, then
\begin{align}
\hat{\omega}^{\{N_{\ell,i}\}_{\ell,i}} \pto \omega.
\end{align}
and
\begin{align}
\label{eq:post-est}
\hat{\pi}^{\{N_{\ell,i}\}}_{\ell,i} := \frac{\hat{\gamma}^{\{N_{\ell,i}\}_{\ell,i}}}{\hat{\omega}^{\{N_{\ell,i}\}_{\ell,i}}}
\end{align}
converges weakly to the measure $\pi(x)$ on $\X$ as $R\to\infty$.
\end{corollary}	
\begin{proof}
Using the same arguments as Lemma~\ref{the:as} with $f(x)=1$ gives
\begin{align}
\hat{\omega}^{\{N_{\ell,i}\}_{\ell,i}} &\pto
\frac{1}{\sum_{\ell \in \mathcal{L}_0} k_{\ell}\left(N_{\ell}\right)} 
\sum_{\ell \in \mathcal{L}_0} k_{\ell}\left(N_{\ell}\right) 
\sum_{i \in \mathcal{I}_{\ell}}
\int_{\X} \iden(x \in A_{\ell,i}) \gamma(dx) \\
&=\frac{1}{\sum_{\ell \in \mathcal{L}_0} k_{\ell}\left(N_{\ell}\right)} 
\sum_{\ell \in \mathcal{L}_0} k_{\ell}\left(N_{\ell}\right) 
\int_{\X} \gamma(dx) \\
&= \omega.
\end{align}
The second result now follows directly from Slutsky's Theorem and Lemma~\ref{the:as}.
\end{proof}

These results firstly convey that if we combine convergent estimators
for the partitioned parts of the overall target, we get a convergent estimator for the target.  
Secondly, it implies that we can similarly combine a number of estimates for the target, which
come from different partitionings. For example, we can combine
a estimate $\hat{\gamma}(x)\iden(x \in A_j)$ for the trivial partition $\{A_j\}$ of $A_j$,
with that given
by combining $\hat{\gamma}(x)\iden(x \in A_{\ell_j})$ 
and $\hat{\gamma}(x)\iden(x \in A_{r_j})$ for the partitioned parts $A_{\ell_j}$ and $A_{r_j}$ where
$A_j=A_{\ell_j}\cup A_{r_j}$,
in a manner that preserves consistency, i.e. we can consistently combine parents estimates with
their children. These results hold independently of how the $k_{\ell}$ are chosen, provided Assumption~\ref{ass:c} holds.
However, the variances of the associated estimates are likely to depend heavily on the choice of $k_{\ell}$ -- we wish to place more weight on the partitionings with lower variance estimates.

A critical point is that the combination of estimators does not require any correction factor for
the number of times that an estimator and a partition were ``proposed'' -- i.e. we do not need to correct for
the fact that more computational resources are provided for some estimates than others or because some
partitions of the space are potentially larger than others.  All such potential factors either cancel out,
or are dealt with by the correct normalization of the truncated proposal.  As such, any strategy on
deciding the partitions or how often a partition is proposed only need satisfy the stated
assumptions to ensure consistency.  We are now thus ready to prove Theorem~\ref{the:est-main} from the main paper as follows, with the Theorem
itself repeated for convenience.

\it*
\begin{proof}
The proof follows using a combination of showing that
Assumptions~\ref{ass:cons},~\ref{ass:c} and~\ref{ass:stop}  are satisfied
and a recursive application
of Lemma~\ref{the:as} and Corollary~\ref{the:post}.

We start by considering $\hat{\gamma}_j (\cdot)$ and $\hat{\omega}_j$ for
a node $j$ whose children are both leaf nodes.  Here  Assumption~\ref{ass:stop}  is trivially satisfied as we have two estimates: the local parent estimate
and the combined child estimate. By construction, the combination of
 a parent node and child node estimates satisfies the partitioning requirements 
 of Assumption~\ref{ass:cons}, while by the final assumption in the theorem, we
 have the required consistency of each of the child and parent
 node estimates in isolation.  Thus Assumption~\ref{ass:cons} is also satisfied.
Assumption~\ref{ass:c} is satisfied through the assumption that each leaf node
is visited infinitely often as the budget becomes arbitrarily large and the fact
that, by construction, $c_j \to 1$ for the parent node as this happens
unless the number of samples used to construct the local parent
estimator also becomes infinitely large, in which case both estimates converge
anyway.

Lemma~\ref{the:as} now tells us that $\hat{\gamma}_j (\cdot) \to \gamma(x)
\iden (x \in A_j)$ and Corollary~\ref{the:post} tells us that
$\hat{\omega}_j \to \int_{x\in A_j} \gamma(x)dx$ and $\hat{\pi}_j(\cdot) \to
\pi(\cdot | x\in A_j)$.  We thus have the Theorem holds for leaf nodes
and all nodes whose children a both leaves.

We can now recursively apply the same logic to show that the Theorem holds
for all nodes in the tree.  Specifically, we have that a node also converges if
both its children nodes convergence, and so by induction all the nodes in the
tree must converge.
%
%
%
%
%
\end{proof}

\begin{remark}
This result can be trivially extended to convergence in probability, $\mathcal{L}^P$ convergence,
and almost sure convergence of the expectation estimates, given the
assumption that both the $\hat{\omega}_j$ and the corresponding unnormalized local 
expectation estimates \[
\hat{\varrho}_j := \frac{1}{N_j}\sum_{n=1}^{N_j} w_j^n f\left(\hat{x}_{j}^n\right)
\]
provide the required convergence.
This follows by simply noting that the arguments in each proof remain
equally valid for $\hat{\varrho}_j$ and for the different forms of convergence.
\end{remark}

\subsection{Discrete Variables}
\label{sec:discrete}

As explained in \S\ref{sec:app:projection}, our method for generating partitions means that they 
are not always disjoint as required by Assumption~\ref{ass:cons}, most notably when $x$ is discrete.
Fortunately, we can still deal with this case by noting that the required properties of Assumption~\ref{ass:cons}
\emph{do} hold in the space of $z_{1:T}$.  This will require no algorithmic changes, but will require additional
consideration in the proof.  In this case we replace Assumption~\ref{ass:cons} with the following
\begin{assumption}
	\label{ass:cons2}	
	Let $z_{1:T} \sim u(z_{1:T})$ be uniformly distributed on the unit hypercube $\Z_T = [0,1]^T$ and let $x = g(z_{1:T})$ have 
	density $q(x)$ and support $x \in \X$, where $q(x)$ is a valid importance sampling proposal for $\gamma(x)$ (see e.g. \citep{owen2013mc}).
	For every independent estimator set $\ell \in \{1,\dots,L\}$, we are given a) a
	partitioning $\{B_{\ell,i}\}_{i \in \mathcal{I}_\ell}$ of $\Z_T$ such that $B_{\ell,i} \cap B_{\ell,j} = \emptyset$
	for $i\neq j$ and $\bigcup_{i \in \mathcal{I}_\ell} B_{\ell,i} = \Z_T$,
	and b) a family $\{\hat{\varphi}^{N_{\ell,i}}_{\ell,i}\}_{i \in \mathcal{I}_\ell}$ of estimated
	measures on $\Z_T$ for all $N\geq 1$:
	\[
	\hat{\varphi}^{N_{\ell,i}}_{\ell,i} \left(\cdot\right) := \frac{1}{N_{\ell,i}} \sum_{n=1}^{N_{\ell,i}}
	w^n_{\ell,i} \, \delta_{\hat{z}_{1:T,\ell,i}^n} \left(\cdot\right)
	\]
	for some random variables $w^n_{\ell,i}$ and $\hat{z}_{1:T,\ell,i}^n$ such that each
	$\hat{\varphi}^{N_{\ell,i}}_{\ell,i}$ converges weakly to the following measure on $\Z_T$ as $N_{\ell,i} \rightarrow \infty$
	\[
	\frac{\gamma(g(z_{1:T})) \iden(z_{1:T} \in B_{\ell,i})u(z_{1:T})}{q(g(z_{1:T}))}.
	\]
	Further each marginal probability estimate converges in probability as follows
	\[
	\hat{\omega}^{N_{\ell,i}}_{\ell,i} := \frac{1}{N_{\ell,i}} \sum_{n=1}^{N_{\ell,i}}
	w^n_{\ell,i} 
	\pto 
	\int_{\Z_T} \frac{\gamma(g(z_{1:T})) \iden(z_{1:T} \in B_{\ell,i})}{q(g(z_{1:T}))} u(dz_{1:T}).
	\]
\end{assumption}

\begin{corollary}\label{the:as2} 
	Let $\hat{\gamma}^{N_{\ell,i}}_{\ell,i}$ denote the pushforward measure of $\hat{\varphi}^{N_{\ell,i}}_{\ell,i}$ (as per 
	$\hat{\gamma}^{N_{\ell,i}}_{\ell,i}$ in Assumption~\ref{ass:cons}), then if Assumptions~\ref{ass:cons2},~\ref{ass:c} and~\ref{ass:stop} hold,
	\begin{align} 
	\label{eq:estimator2} 
	\hat{\gamma}^{\{N_{\ell,i}\}_{\ell,i}} :=  \frac{1}{\sum_{\ell=1}^{L} k_{\ell} \left(N_{\ell}\right)} 
	\sum_{\ell=1}^{L} k_{\ell}\left(N_{\ell}\right) \sum_{i \in \mathcal{I}_{\ell}} \hat{\gamma}^{N_{\ell,i}}_{\ell,i} 
	\end{align}
	converges weakly to the measure $\gamma(x)$ on $\X$ as $R\to\infty$.
\end{corollary}
\begin{proof}
As per Lemma~\ref{the:as}, the estimates for $\ell \notin \mathcal{L}_0$ need not converge
but do not affect the final estimate.  We again demonstrate the result by considering
an arbitrary continuous function $f : \X \to \real$ for which we have
\begin{align*}
\int f(x)\, \hat{\gamma}^{\{N_{\ell,i}\}_{\ell,i}}(d x)
&=         
\int
f(x)
\frac{1}{\sum_{\ell=1}^{L} k_{\ell}\left(N_{\ell}\right)} 
\sum_{\ell=1}^{L} k_{\ell}\left(N_{\ell}\right) \sum_{i \in \mathcal{I}_{\ell}} \hat{\gamma}^{N_{\ell,i}}_{\ell,i}(d x) 
        \\
        &
        =
        \frac{1}{\sum_{\ell=1}^{L} k_{\ell}\left(N_{\ell}\right)} 
        \sum_{\ell=1}^{L} k_{\ell}\left(N_{\ell}\right) 
        \sum_{i \in \mathcal{I}_\ell} 
        \int f(x) \hat{\gamma}^{N_{\ell,i}}_{\ell,i}(d x) \\
&
=
\frac{1}{\sum_{\ell=1}^{L} k_{\ell}\left(N_{\ell}\right)} 
\sum_{\ell=1}^{L} k_{\ell}\left(N_{\ell}\right) 
\sum_{i \in \mathcal{I}_\ell} 
\int f(g(z_{1:T})) \hat{\varphi}^{N_{\ell,i}}_{\ell,i}(dz_{1:T})
\end{align*}
which using Assumptions~\ref{ass:cons2}
and~\ref{ass:c} converges as $R\to\infty$ to
\begin{align*}
	 \frac{1}{\sum_{\ell \in \mathcal{L}_0} k_{\ell}\left(N_{\ell}\right)}
	 	&\sum_{\ell \in \mathcal{L}_0}  k_{\ell}\left(N_{\ell}\right) 
	 	\sum_{i \in \mathcal{I}_\ell}
	 	\int \frac{f(z_{1:T})\gamma(g(z_{1:T})) \iden(z_{1:T} \in B_{\ell,i})}{q(g(z_{1:T}))}u(dz_{1:T}) \\
	 		 &=\frac{1}{\sum_{\ell \in \mathcal{L}_0} k_{\ell}\left(N_{\ell}\right)}
	 		 \sum_{\ell \in \mathcal{L}_0}  k_{\ell}\left(N_{\ell}\right) 
	 		 \int \frac{f(z_{1:T})\gamma(g(z_{1:T}))}{q(g(z_{1:T}))}u(dz_{1:T}) \\
	 &  =\int \frac{f(z_{1:T})\gamma(g(z_{1:T}))}{q(g(z_{1:T}))}u(dz_{1:T}) 
	   =\int \frac{f(x)\gamma(x)}{q(x)}q(dx) 
	 =
	 \int f(x)  \, \gamma(dx)
\end{align*}
as required.
\end{proof}

Given this corollary, we can now trivially extend Theorem~\ref{the:est-main} to the setting where
Assumption~\ref{ass:cons2} holds instead of Assumption~\ref{ass:cons} using the same arguments.
%

\section{Setting the Child Preference Factors}
\label{sec:app:child-pref}

The child preference factors $c_j$ represent a relative weight given to the estimate
from the child nodes in our combined estimator.  In the absence of other information, it
would thus be natural to set $c_j = \frac{M_j-N_j}{M_j}$ where $M_j$ is the total number
of traversals (including running inference at the parent) and $N_j$ is the number of times
inference has been run at the parent node, such that the estimates are weighted
in proportion to the number of component samples.  However, we also expect the \emph{per-sample efficiency}
of the child estimate to be better than the parent because of the adaptation provided
by the inference tree.  Therefore, we want to give more preference to the child estimates.
To do this, we employ the simple, but effective, heuristic of scaling the number of child traversals 
as follows
\begin{align}
	c_j = \frac{\lambda^{(\expt[d_{\text{ch}}]-d_j)} (M_j-N_j)}{N_j + \lambda^{(\expt[d_{\text{ch}}]-d_j)} (M_j-N_j)}
\end{align}
where $d_j$ is the depth of node $j$ in the tree and $\expt[d_{\text{ch}}]$ is the average depth
of the child subtrees.  Here is $\lambda \in [1,\infty)$ is preference parameters and can be interpreted
as how many times more efficient we expect the $d_{j+1}$-th layer to be than the $d_j$-th layer.
We use $\lambda = 1.2$ as a default.  In the context of the notation of the main paper we thus have
\begin{align}
c_j = \frac{M_j-N_j}{M_j}\chi_j \quad \text{where} \quad \chi_j = \frac{\lambda^{(\expt[d_{\text{ch}}]-d_j)} M_j}{N_j + \lambda^{(\expt[d_{\text{ch}}]-d_j)} (M_j-N_j)}
\end{align}
is our correction factor.

\section{Estimates for Empirical Variance and Effective Sample Size}
\label{sec:app:ess}

When calculating terms such as the effective sample size (ESS)~\citep{owen2013mc}, we need to
take care about the fact that our traversal strategy implies additional implicit weights through the $N_j$ and
$c_j$.  In short, our ``expected squared weight'' should not be simply calculated using $\frac{1}{N_j} \sum_{n=1}^{N} (w_i^n)^2$
but instead using the scheme we now introduce.  Given this expected squared weight estimator, a number
of useful estimators will follow naturally.

We start by introducing an alternative formulation of the combined marginal likelihood estimate of
a node as follows
\begin{align}
\hat{\omega}_j &= (1-c_j) \frac{1}{N_j} \sum_{n=1}^{N_j} w_j^n + c_j \left(\hat{\omega}_{\ell_j}+\hat{\omega}_{r_j}\right)
= \frac{1}{M_j} \sum_{m=1}^{M_j} w_m k_m \frac{M_j}{N_{j(m)}}
\end{align}
where $M_j$ is the number of times the node has been traversed, $\{w_m\}_{m=1}^{M_j}$ is the union of all the weights
from the current node and its decedents, $k_m$ is a child preference weight associated with sample $m$ (e.g. $(1-c_j)$
for a sample form the current node local estimate, $c_j(1-c_j)$ for a sample from the local estimate of a child if that node is an internal node, etc.),
and $N_{j(m)}$ is the number of samples that have been generated locally at the node that generated sample $m$.
We thus see that the true sample weights in our combined estimator are $w_m k_m M_j/N_{j(m)}$ and so our
estimator for the squared weight is
\begin{align}
\hat{\zeta}^2_j :=& \frac{1}{M_j} \sum_{m=1}^{M_j} \left(w_m k_m \frac{M_j}{N_{j(m)}}\right)^2
= M_j \sum_{m=1}^{M_j} \left(\frac{w_m k_m}{N_{j(m)}}\right)^2 \nonumber \\
=& M_j \left((1-c_j)^2 \frac{1}{N_j^2} \sum_{n=1}^{N_j} \left(w_j^n\right)^2 
+ c_j^2 \left(\frac{\hat{\zeta}_{\ell_j}^2}{M_{\ell_j}}+ \frac{\hat{\zeta}_{r_j}^2}{M_{r_j}}\right)\right).
\end{align}
$\hat{\zeta}_{j}^2/M_{j}$ can be propagated in a similar fashion to other
estimates, allowing $\hat{\zeta}^2_j$ to be estimated at any node.

Given $\hat{\zeta}^2_j$, we can straightforwardly construct various useful estimators.  For example, 
the Monte Carlo estimator for the variance of the weight produced by a given traversal is given by
\begin{align}
\label{eq:sigma}
\sigma_j^2 := \frac{M_j}{M_j-1} \left(\hat{\zeta}^2_j-\hat{\omega}^2_j\right)
\end{align}
where the first term is Bessel's correction.  The ESS, on the other hand, is
\begin{align}
\label{eq:ess}
\text{ESS}_j := \frac{M_j\hat{\omega}^2_j}{\hat{\zeta}^2_j}.
\end{align}

\section{Derivation of the Pure-Exploitation Target}
\label{sec:app:exploit-target}

For this derivation, it will be convenient to first consider the case where the children we are deciding between
are both leaf nodes and that there is some arbitrary (unknown) target function $f$, such that combined child estimate
(not including the parent) is given by
\begin{align}
\label{eq:ch-comb}
\hat{\mu}_{\text{ch}} := \frac{1}{N_{\ell}} \sum_{n=1}^{N_{\ell}} w_{\ell}^n \hat{f}_{\ell}^n + 
\frac{1}{N_{r}} \sum_{n=1}^{N_{r}} w_{r}^n \hat{f}_{r}^n 
\end{align}
where $\hat{f}_{i}^n := f(\hat{x}_{i}^n)$.  Now the mean squared error (MSE) of our
estimator decomposes in the standard manner \[
\expt [(\hat{\mu}_{\text{ch}}-\mu_j)^2] = 
\var[\hat{\mu}_{\text{ch}}]+(\expt [(\hat{\mu}_{\text{ch}}-\mu_j)])^2
\]
where the second term is the biased squared and all terms are implicitly conditioned on $N_{\ell}$
and $N_r$.  Though the finite sample bias of our estimator
is difficult to assert, we know that it vanishes as $N_{\ell},N_r \to \infty$ and, due to the
central limit theorem, we can safely assume this happens faster than the standard deviation vanishes.
Thus asymptotically, we only need to consider the variance to minimize the MSE.  Now, invoking 
the conditional independence given $N_{\ell}$ and $N_r$ of each child estimator and each sample within those
estimators, we have 
\begin{align*}
\var[\hat{\mu}_{\text{ch}}] &= \var\left[\frac{1}{N_{\ell}} \sum_{n=1}^{N_{\ell}} w_{\ell}^n \hat{f}_{\ell}^n\right] + 
\var \left[\frac{1}{N_{r}} \sum_{n=1}^{N_{r}} w_{r}^n \hat{f}_{r}^n\right] \\
&= \frac{1}{N_{\ell}}\var\left[w_{\ell}^1 \hat{f}_{\ell}^1\right]
+\frac{1}{N_r}\var\left[w_{\ell}^1 \hat{f}_{\ell}^1\right].
\end{align*}
Using the stratified sampling results of, for example,~\citep{carpentier2015adaptive}, it
is straightforward to show that the subsequent optimal strategy is to set
\begin{align*}
N_{\ell} \propto \sqrt{\var\left[w_{\ell}^1 \hat{f}_{\ell}^1\right]}
\quad \text{and} \quad
N_{r} \propto \sqrt{\var\left[w_{r}^1 \hat{f}_{r}^1\right]}.
\end{align*}
Now assuming that the weights and evaluations are independent (remembering that we are considering an arbitrary $f$) we have
\begin{align*}
\var\left[w_{\ell}^1 \hat{f}_{\ell}^1\right]
&=\expt\left[\left(\hat{f}_{\ell}^1\right)^2\right]\expt\left[\left(w_{\ell}^1\right)^2\right]-
\left(\expt\left[\hat{f}_{\ell}^1\right]\right)^2\left(\expt\left[w_{\ell}^1\right]\right)^2\\
&= \expt\left[\left(\hat{f}_{\ell}^1\right)^2\right]\var\left[w_{\ell}^1\right]
+\var\left[\hat{f}_{\ell}^1\right]\left(\expt\left[w_{\ell}^1\right]\right)^2 \\
&= \var\left[\hat{f}_{\ell}^1\right] \left(\var\left[w_{\ell}^1\right]
\left(1+\frac{\left(\expt\left[\hat{f}_{\ell}^1\right]\right)^2}{\var\left[\hat{f}_{\ell}^1\right]}\right)
+\left(\expt\left[w_{\ell}^1\right]\right)^2
\right)
\end{align*}
and similarly for $\var\left[w_{r}^1 \hat{f}_{r}^1\right]$.  We thus
have that the optimal strategy is to set (using $\sigma$ to denote standard deviation)
\begin{align}
\begin{split}
 N_{\ell} \propto \frac{\sigma \left[\hat{f}_{\ell}^1\right]}{\sigma \left[\hat{f}_{\ell}^1\right]
	+\sigma \left[\hat{f}_{r}^1\right]} 
\sqrt{\left(\var\left[w_{\ell}^1\right]
	\left(1+\frac{\left(\expt\left[\hat{f}_{\ell}^1\right]\right)^2}{\var\left[\hat{f}_{\ell}^1\right]}\right)
	+\left(\expt\left[w_{\ell}^1\right]\right)^2
	\right)}
\end{split}.
\end{align}
Here the first term depends only on the unknown target function.  Though one might want to potentially
postulate a particular dependence of $\sigma \left[\hat{f}_{\ell}^1\right]$ on the relative volume
of the nodes, we will just presume the ratio is unknown and conservatively set it to $1$, falling in line
with standard approaches for Bayesian inference where we aim to sample in proportion to the posterior, rather
than artificially producing more samples in larger areas of the space to account for the potential
of higher variation in the target function.

The second term depends only on statistics of the sample weights and the ratio 
$\left(\expt\left[\hat{f}_{\ell}^1\right]\right)^2 / \var\left[\hat{f}_{\ell}^1\right]$.  As $f$ is 
unknown, we also do not know this ratio.  However, we do know it must vary between $0$ (when $\expt\left[\hat{f}_{\ell}^1\right]=0$
or $\var\left[\hat{f}_{\ell}^1\right]\rightarrow\infty$)
and $\infty$ (when $\var\left[\hat{f}_{\ell}^1\right]=0$, i.e. the function is flat).
These two respective extremes give
\[
N_{\ell} \propto \sqrt{\left(\var\left[w_{\ell}^1\right]
	+\left(\expt\left[w_{\ell}^1\right]\right)^2\right)} = \sqrt{\expt [\left(w_{\ell}^1\right)^2]} \quad \text{and} \quad 
N_{\ell} \propto \sigma [w_{\ell}^1].
\]  
The latter of these corresponds to the optimal strategy for estimating the marginal
likelihood,  as would be expected from considering the stratified sampling results of~\cite{carpentier2015adaptive}
applied to estimating $\expt[w]$.  However, this strategy gives no consideration of the need to
produce samples from areas of high posterior density to capture possible variations in the target function
and so is highly inappropriate.  Assuming the former extreme is more conservative and spends time sampling in regions
of high probability mass and also those of high weight uncertainty.

Rather than taking a particular extreme, we treat 
\[
\kappa := \frac{\left(\expt\left[\hat{f}_{\ell}^1\right]\right)^2}{\var\left[\hat{f}_{\ell}^1\right]}, \quad \kappa \in [0,\infty]
\]
explicitly as a parameter of the traversal algorithm, where higher values of $\kappa$
give more emphasis to estimating the marginal likelihood and to accurate prediction of expectations
of smoothly varying functions, while lower values of $\kappa$ give more emphasis to sampling regions
in proportion to their marginal probabilities.  We note the interesting, and perhaps counter-intuitive,
result that even when $\kappa$ is its
minimum possible value, the optimal traversal strategy is still not to sample in proportion to marginal
probability, except in the special case where the variance of the weights is zero.  

Thus far we have omitted the fact that we eventually want a normalized estimator.  We deal with the former
by noting that we intend to \emph{separately} propagate the unnormalized estimate and the marginal
likelihood estimate.  Thus, except at the root node, our aim is to propagate low variance estimates
of both, rather than simply low variance estimates of the ratio.  Though we do not do further analysis
to assess this, we choose by default to set $\kappa=1$, to reflect
the fact that we thus always explicitly care about the marginal likelihood estimate.

We have also thus far omitted the fact that we need to calculate traversal strategies when the 
children are not leaves.  Here we can use the same analysis but need to replace $\expt\left[w_{\ell}^1\right]$
and $\var\left[w_{\ell}^1\right]$ with appropriate combined estimators.  For the former, we
can simply use $\hat{\omega}_{\ell}$.  For the latter, we need a notion of a ``single-traversal'' variance in the
marginal likelihood estimate.  Such a metric was derived as $\hat{\sigma}_{\ell}^2$ in Appendix~\ref{sec:app:ess}.
We thus arrive at our derivation of the unnormalized exploitation reward of node $\ell$ as
\begin{align}
\label{eq:app:exploit-reward}
\hat{\tau}_{\ell} := \sqrt{\hat{\omega}_{\ell}^2+(1+\kappa) \hat{\sigma}_{\ell}^2}.
\end{align}

\section{Additional Density Estimation Details}
\label{sec:app:den}

\subsection{Additional Heuristics}
\label{sec:app:true-lik}

Even though we cannot calculate it, we know that there is maximum possible log weight for
each node, namely 
\[
\log w_j^* = \max_{z_{1:T} \in B_j} \log \gamma(g(z_{1:T})) + \log \lVert B_j \rVert - \log q(g(z_{1:T})).
\]
Consequently, our density estimator (which is defined on the full real line)
will typically slightly overestimate the probability of a sampling falling above the threshold.  In particular, if there
is a large number of samples at the node and we are only using a simple density estimator for $\psi$, we may
continue to except to exceed the threshold even when previous samples suggest a saturation below the threshold.

Let $\text{e}(T)$ to denote the event $\left\{\max(w_j^{1:T})>w_{\text{th}}\right\}$, i.e. the event that one of $T$ independent samples
exceeds the threshold if we draw $T$ samples.  We now have
$P(\text{e}(T)) = 1-(1-\Psi(\log w_{\text{th}}))^T$. We
can further condition this on the event that we have not already seen the threshold exceeded
using the likelihood
$P(\neg \text{e}(N_j) | \text{e}(T))$. To define this, we introduce
 an additional parameter
$\log w_{\text{gap}}$ and define our likelihood to condition on the fact that none of our $N_j$ samples
fall above $\log w_{\text{th}}$ with $\Psi$ truncated at 
$\log w_{\text{tr}}:=\log w_{\text{th}}+\log w_{\text{gap}}$ to reflect the fact
that the true log weights are bounded, giving
\begin{align*}
P(\neg \text{e}(N_j) | \text{e}(T)) = \left(\frac{\Psi(\log w_{tr})-\Psi(\log w_{\text{th}})}{\Psi(\log w_{tr})}\right)^{N_j},
\end{align*}
with Bayes' rule in turn yielding
\begin{align}
P(\text{e}(T) | \neg \text{e}(N_j))
= \frac{\left(1-(1-\Psi(\log w_{\text{th}}))^T\right)P(\neg \text{e}(N_j) | \text{e}(T))}
{\left(1-(1-\Psi(\log w_{\text{th}}))^T\right)P(\neg \text{e}(N_j) | \text{e}(T)) + (1-\Psi(\log w_{\text{th}}))^T}.
\end{align}
The full definition of $\hat{p}^s_j$ actually used is then given by
\begin{align}
\hat{p}^s_j := \left(1-c_j\right)\frac{P(\text{e}(T) | \neg \text{e}(N_j))}{\text{ESS}_j}
+c_j \left(\hat{p}^s_{\ell_j}+\hat{p}^s_{r_j}-\hat{p}^s_{\ell_j}\hat{p}^s_{r_j}\right).
\end{align} 

\subsection{Additional Intuition and Parameters}
\label{eq:app:int}

At first it might seem counter intuitive to include an ESS scaling term in $\hat{p}^s_j$ as
a classic failure case for the ESS as a performance metric is if there multiple modes.  However,
the scenario where the \emph{local} estimate has a high ESS and multiple modes is expected
to be rare.  Instead, one will typically have a low local ESS for any node with multiple
modes but it may have children with a high ESS giving it a high combined ESS estimate.  In these
cases, the combined significant probability estimate $\hat{p}^s_j$ should still be high if there
is any descendant $i$ with a high $\hat{p}^s_i$ and a low $\text{ESS}_i$.  Thus in practice, scaling by the
ESS does not cause the high nodes in the tree to miss multiple mode cases, while providing
a more reliable metric for nodes low down in the tree.

In our approach, $T$ and $\log w_{\text{gap}}$ constitute fixed parameters which we set to
$1000$ and $10$ respectively as default.  On the other hand, $w_{\text{th}}$ naturally
needs to change as the training progresses.  We make the simple choice of setting $w_{\text{th}}$
to being the highest weight generated at any node, scaled to adjust for differences $\lVert B_j \rVert$
makes to the weight. An unfortunate feature of this choice is that whenever
the MAP estimate changes, the $\hat{p}^s_j$ for all nodes must be updated.  However, the regularity
that this occurs diminishes with the number of iterations, such that it, in practice, does not
lead to an increasing per-iteration computational cost as the tree is run longer.

\section{Additional Details on Refinement Strategy}
\label{sec:app:refine}

To define our entropy metric more precisely, recall that the entropy of a continuous uniform distribution $U(s_1, s_2)$ is
\begin{align}
\textsc{Entropy}(U(s_1,s_2)) 
&= - \int_{s_1}^{s_2} {q(z_t)\log{q(z_t)}}\,\mathrm{d}z_t= \ln(s_2 - s_1).
\end{align}
Assume that we propose a split at a point $s \in (s_1, s_2)$,
and that we will later go to the left of this split with 
a probability $P_{\ell}$ and to the right with
a probability $P_r = 1 - P_{\ell}$. This splitting and the traversal strategy
give rise to a proposal of a mixture of two uniform distributions
that has the following density
\begin{align}
q_s(z_t) =
\begin{cases}
d_{\ell}          & \text{if } z_t < s, \, \text{where} \, d_{\ell} = \frac{P_{\ell}}{s-s_{1}}\\
d_{r}          & \text{otherwise}, \, \text{where} \, d_{r} = \frac{1- P_{\ell}}{s_{2}-s}
\end{cases}
\end{align}
The entropy of this proposal is:  
\begin{align}
\textsc{Entropy}(p_a) 
& = - \int_{s_1}^s d_{\ell}\log{d_{\ell}}\,\mathrm{d}z_t - \int_{s}^{s_2} d_r\log{d_r}\,\mathrm{d}z_t = P_{\ell} \log\frac{1}{d_{\ell}} + P_r\log\frac{1}{d_r}\, .
\end{align}
We can now use our empirical estimates $\hat{\omega}_{\ell} = \frac{1}{N}\sum_{i =1}^{N} \mathbb{I} (z_{t,i} \in B_{\ell})w_i $ and similarly $\hat{\omega}_{r} = \frac{1}{N}\sum_{i =1}^{N} \mathbb{I} (x_i \notin B_{\ell})w_i$ 
to define our entropy metric as
\begin{align*}
\textsc{Entropy}(p_s) & 
= 
\hat{P}_{\ell} \log\frac{s-s_1}{\hat{P}_{\ell}} + 
\hat{P}_r \log\frac{s_2-s}{\hat{P}_r}\, \quad
\text{where} \quad \hat{P}_{\ell} = \frac{\hat{\omega}_{\ell}}{\hat{\omega}_{\ell}+\hat{\omega}_{r}} \quad \text{and} \quad
\hat{P}_{r} = \frac{\hat{\omega}_{r}}{\hat{\omega}_{\ell}+\hat{\omega}_{r}},
\end{align*}
which is trivially equivalent to the loss given in~\eqref{eq:split-criterion} up to
a normalization constant.
We then choose the split $s^* =  \operatorname{arg\,min}_s \textsc{Entropy}(p_s)$ where the
minimization is over our randomly sampled candidate splits.

As a minor additional heuristic aimed at avoiding splits where a small but significant proportion
of the tail is contained within one child, we do not in practice use $s^*$ directly, instead reducing
the size of the child with lower probability mass by 25\%.

After choosing the best split among all candidates and separating the space in to $B_{\ell}$ and $B_r$, we run inference restricted to $B_{\ell}$ and $B_r$ separately. Then we compare the empirical estimates of the marginal likelihood for each child using a t-test, which shows how likely the results are samples from two different distributions. If the p-value is small, it suggests the split is meaningful. In that case, we accept the split, creating two new child nodes and converting the current leaf node to a discriminant note. Otherwise, we discard the split and combine the samples, adding them to the estimate of the current node. When the node is revisited, new splits are suggested and the process continues in the same way.

\section{Additional Experimental Details}
\label{sec:app:exp}

\subsection{Gaussian Mixture Model}

For the GMM experiment, the IT parameters were set as
$\kappa = 1$, $\beta = 0.1$, and $c_j$ as per Appendix~\ref{sec:app:child-pref} with $\lambda$.  Denoting $\rho$ as the
proportion of total iterations run thus far, the annealing
parameters were given schedules of $\delta(\rho) = \frac{1}{2} \left(1+\tanh \left(20\left(0.9-\rho\right)\right)\right)$
and $\alpha(\rho) = \frac{5}{8} \left(1+\tanh \left(25\left(0.95-\rho\right)\right)\right)$.
We further fixed $\beta=0$ for the last 25\% of the iterations to reflect the fact that, because
we are carrying out inference rather than optimization, we want to spend part of our sample
budget more directly exploiting the learned tree.

Our main baseline, PI-MAIS~\citep{martino2017layered} is a state-of-the-art
adaptive importance sampling algorithm that runs a number of independent MCMC chains targeting the joint distribution
and then uses the locations of these chains to, at each iteration, construct a mixture of Gaussian proposal 
distribution, with each component centered on the location of one of the chains.  We used $N=100$ such
chains and proposed $M=15$ samples from each chain at each iteration, noting that the algorithm requires 
$N(M+1)$ target evaluations.  We further used an random walk kernel with covariance $0.0001 {I}$ 
for each of the MCMC chains, while each proposal component is taken as an isotropic Gaussian with 
covariance $0.01{I}$.

\subsection{Chaotic Dynamics Model}

This model comprises of an extended Kalman
filter defined as 
\begin{align*} 
x_0 &\sim  \, \mathcal{N} \left(0, {I}\right) \\ 
f_t (x_{t}|x_{t-1}) &=  \, A \left(x_{t-1}, \theta\right)+\upsilon_{t-1}, \quad \upsilon_{t-1} \sim \mathcal{N} \left(0, 0.01 I\right) \\ 
g_t (y_{t} | x_{t}) &=  \, C x_{t}+\varepsilon_{t}, \quad \varepsilon_{t} \sim \mathcal{N} \left(0, 0.2 I\right). 
\end{align*}
where $C$ is a known $K \times 3$ matrix. The transition function $A \left(\cdot,\theta\right)$ dictates 
the underlying dynamics with parameters
$\theta$.  We will assume that the form of $A$ is known but not the parameters.  Namely, we consider
the example where the dynamics correspond to the Pickover attractor defined as
\begin{align*} 
x_{t,1} = & \; \sin \left(b x_{t-1,2}\right)-\cos\left(a x_{t-1,1}\right)x_{t-1,3} \\ 
x_{t,2} = & \; \sin \left(d x_{t-1,1}\right)x_{t-1,3}-\cos\left(c x_{t-1,2}\right) \\ 
x_{t,3} = & \; \sin \left(x_{t-1,1}\right)
\end{align*}
where $\theta = (a, b, c, d)$.  We finish the model by defining
the prior on each dynamics parameter to be a uniform over $[-\pi,\pi]$.  A synthetic dataset
$y_{1:200}$ was generated by fixing $b=-2.3$, $a=2.5$, $d=-1.5$,  $c=1.25$, $K=20$,
and drawing each column of $C$ from a symmetric Dirichlet distribution with concentration $0.1$.

Our main baseline method was PMMH~\cite{andrieu2010particle}, a pseudo-marginal method where one runs an MCMC sampler
targeting $p(\theta | y_{1:T})$ but with the likelihood evaluation in the MH acceptance
step replaced with the unbiased ML estimate produced by an SMC sweep.
For this, we use isotropic random walk proposal with a covariance of $0.0004 I$.  For the
SMC sweeps, we used 500 particles and the bootstrap proposal.

For this experiment, the same IT parameters were used as the GMM experiment, with the exception that
we changed the annealing schedules to match the lower number
of iterations, setting $\delta(\rho) = \frac{1}{2} \left(1+\tanh \left(4\left(0.7-\rho\right)\right)\right)$
and $\delta(\alpha) = \frac{5}{8} \left(1+\tanh \left(10\left(0.8-\rho\right)\right)\right)$.

\section{PMMH Sample Paths}
\label{sec:app:pmmh}

\begin{figure*}[t!]
	\centering
	\includegraphics[width=0.45\textwidth]{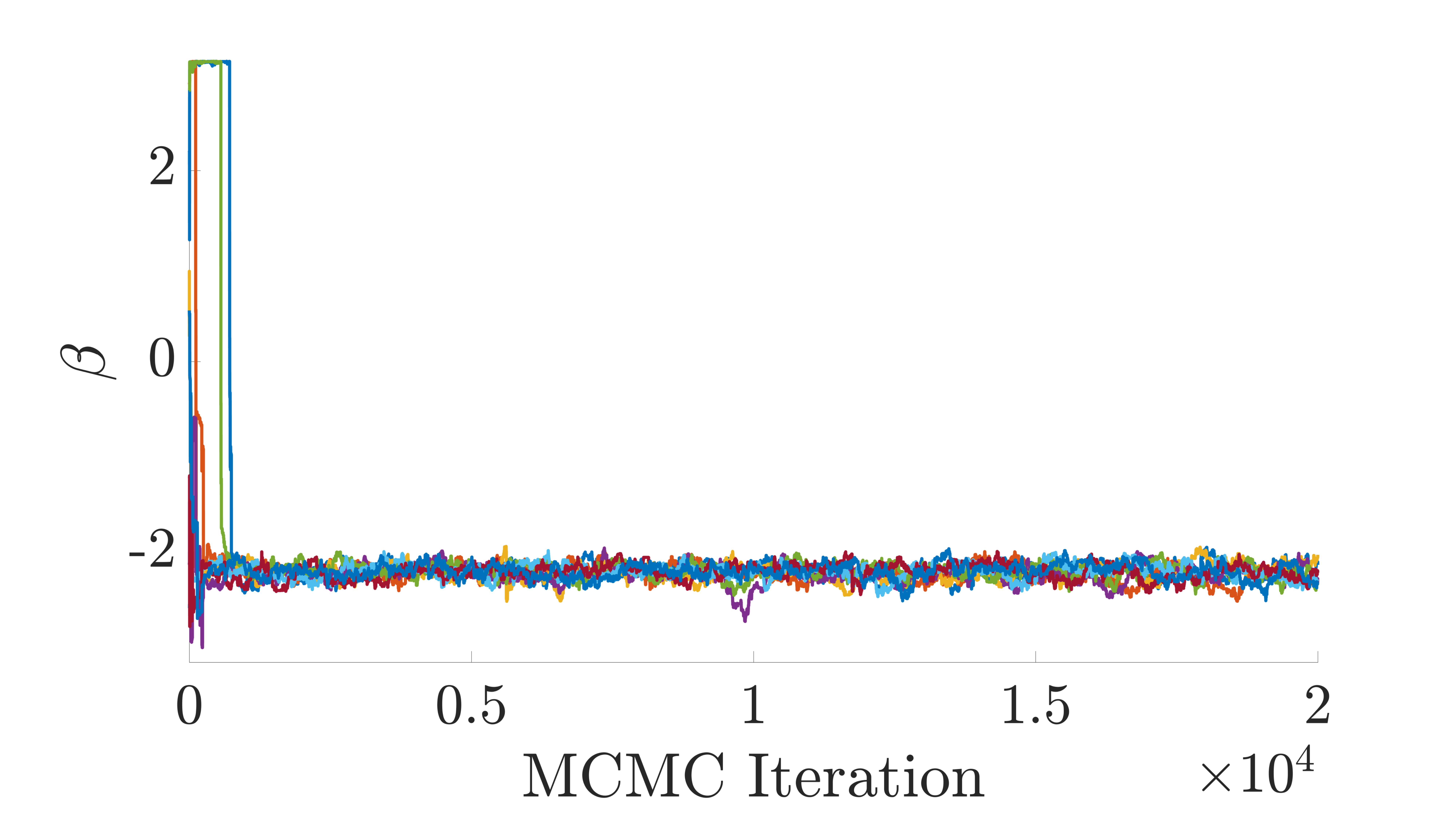} ~~~ 
	\includegraphics[width=0.45\textwidth]{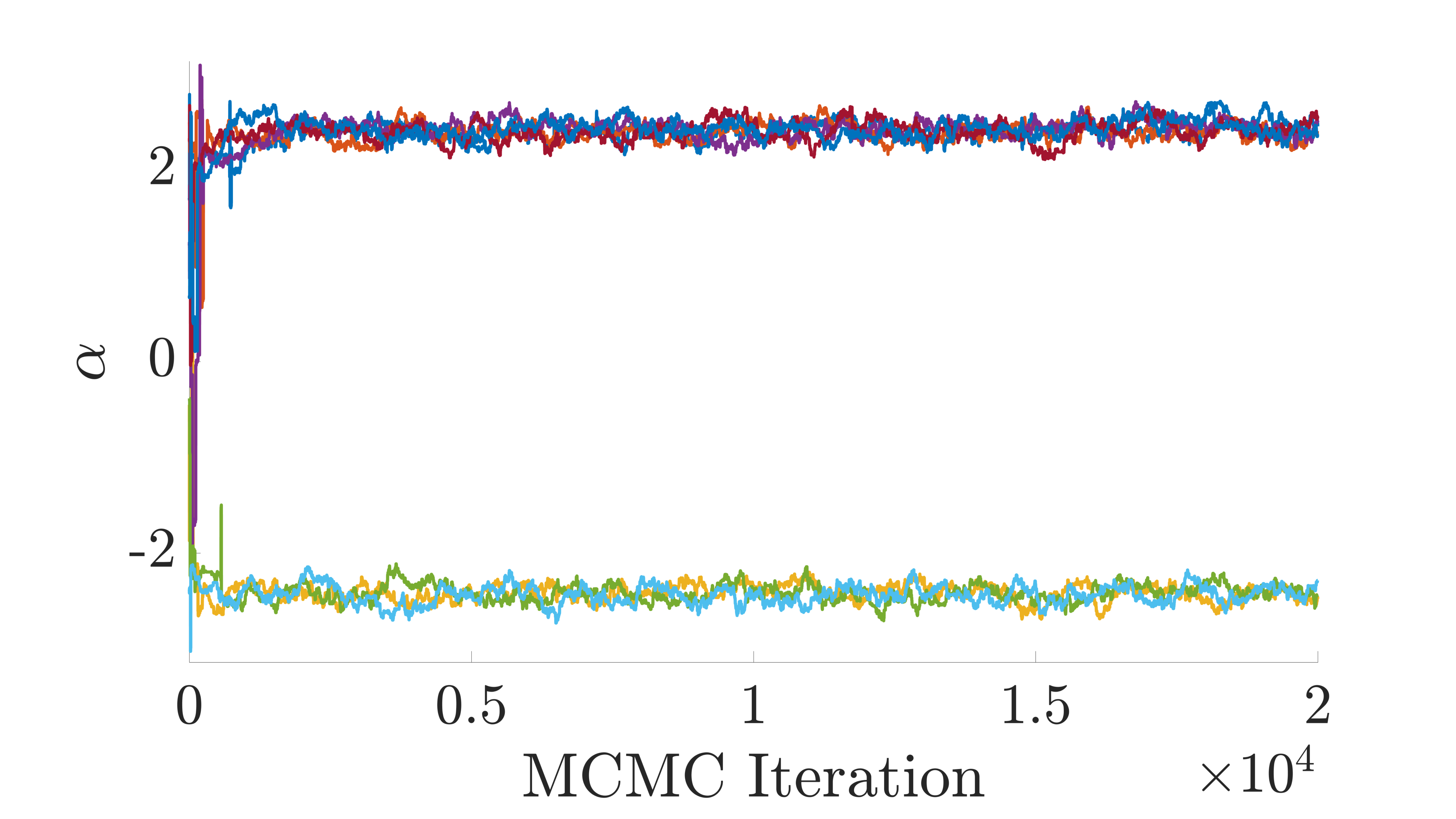} \\
	\includegraphics[width=0.45\textwidth]{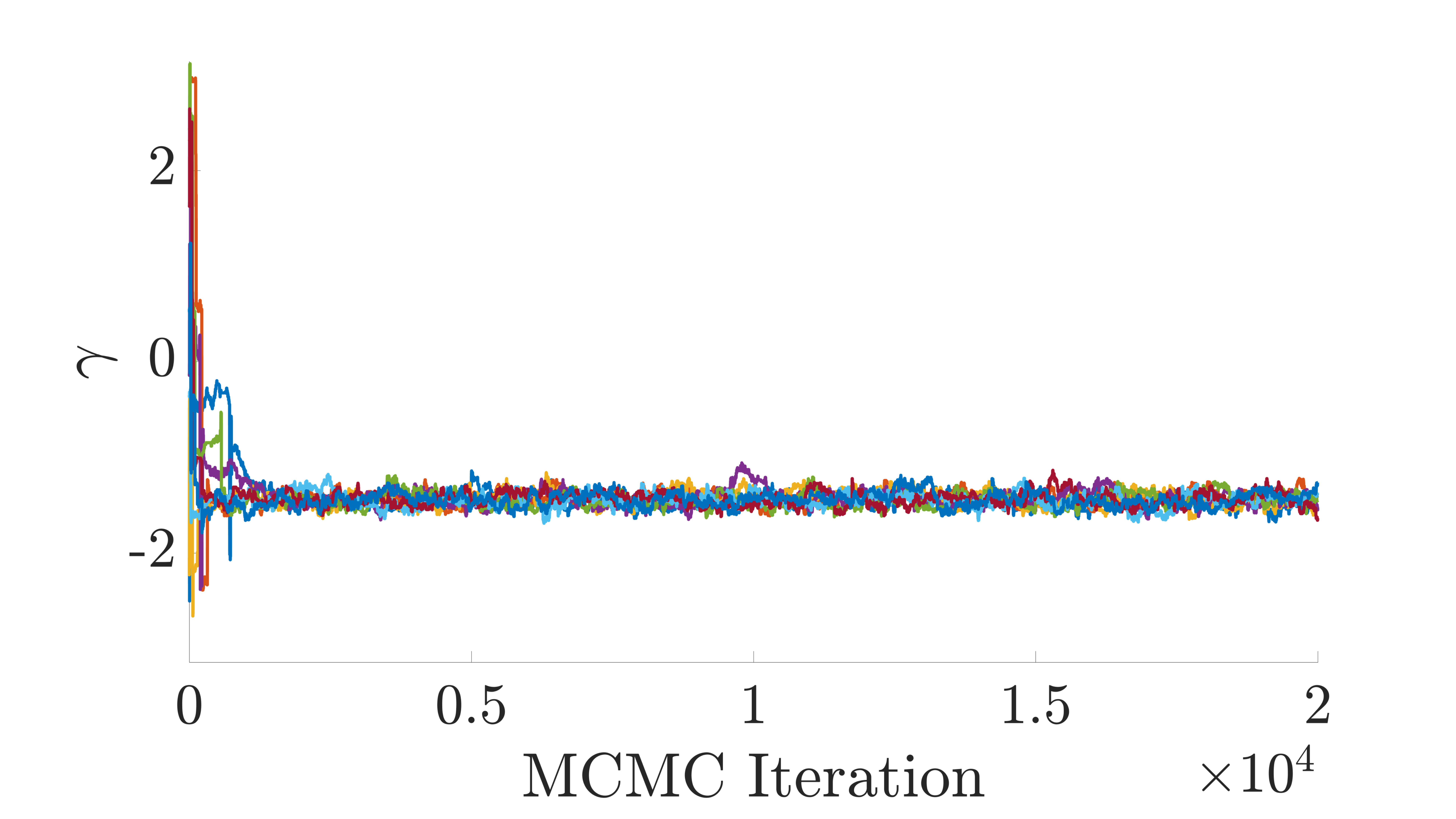} ~~~ 
	\includegraphics[width=0.45\textwidth]{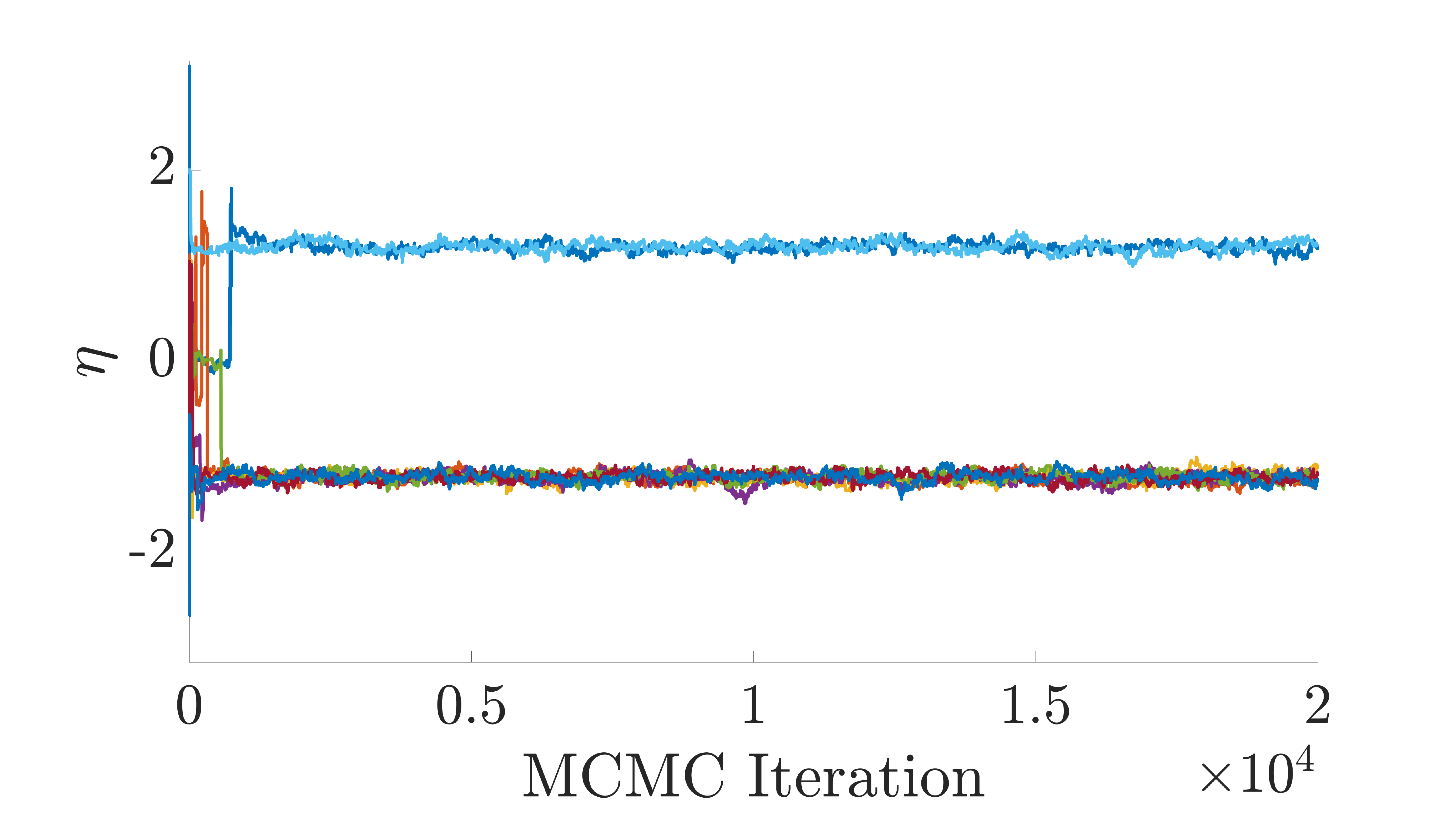} \\	
	\caption{Sample paths of PMMH.
		\label{fig:pmmh_eta}}
\end{figure*}

In the main paper, we only showed results from a single run of PMMH.
To demonstrate that PMMH fails to move between modes in any of the runs, we now
plot the individual sample paths as shown in Figure~\ref{fig:pmmh_eta}.  We see
that for the parameters with multiple modes, $\alpha$ and $\eta$, the PMMH sampler never
moves between the modes.  Thus in all runs we see PMMH was only able to pick
up a single mode.

\section{ITs for Integration}	
\label{sec:app:det}

Most adaptive sample schemes only look to approximate the posterior in the most
accurate way, ignoring the fact that there might be a known function $f$ which we 
are trying to estimate the expectation of, namely $\expt_{\pi(x)} \left[f(x)\right]$.
Clearly, this is inferior when $f$ is known, as it ignores the fact that $f$ may have higher
variability in some regions than others, such that the accuracy in those regions
is more impactful on the error in the overall estimate.  As well as being used as an
adaptive inference algorithm, ITs are also capable of operating in this integration setting
as we now demonstrate.

The integration setting for ITs varies primarily in the traversal strategy.  In
Section~\ref{sec:app:exploit-target}, we indirectly showed that the optimal exploitation strategy
for the known $f$ case is 
\[
N_{\ell} \propto \sqrt{\var\left[w_{\ell}^1 \hat{f}_{\ell}^1\right]}
\quad \text{and} \quad
N_{r} \propto \sqrt{\var\left[w_{r}^1 \hat{f}_{r}^1\right]},
\]
a result that has been previously noted by, for example,~\cite{carpentier2015adaptive} in
the stratified sampling literature.  Unlike where $f$ is unknown, 
$\var\left[w_{\ell}^1 \hat{f}_{\ell}^1\right]$ here is a term we can directly
estimate in the same way as $\sigma_j^2$ (see Section~\ref{sec:app:ess}).  
Defining $\hat{s}_j^2$ as the equivalent of $\hat{\sigma}_j^2$ when replacing the
weights with $w_j^n \hat{f}_{j}^n$, this gives that exploitation target
is simply
\begin{align}
\tau_j^{\text{int}} = \hat{s}_j.
\end{align}
Unfortunately, our exploration strategy using density estimation does not translate
so simply to the integration setting.  We thus leave developing an analogous approach
to future work, and simply set
\begin{align}
u_{\ell_j} = \frac{1}{M_{\ell_j}}\left(\left(\frac{s_{\ell_j}}{s_j}\right)^{(1-\alpha)}
+ \beta \frac{\lVert B_{\ell_j} \rVert}{\lVert B_{j} \rVert} \frac{\log M_j}{\sqrt{M_{\ell_j}}}\right).
\end{align}
This target is now analogous to that discussed in~\cite{carpentier2015adaptive}
and so their regret analysis should still apply.  

To demonstrate that IT are still useful in this integration setting even without
a principled exploration term in the traversal target, we conducted an experiment based on a network model.  Here our network
\begin{wrapfigure}{r}{0.5\textwidth}
	\centering
	\vspace{-10pt}
	\includegraphics[width=0.5\textwidth]{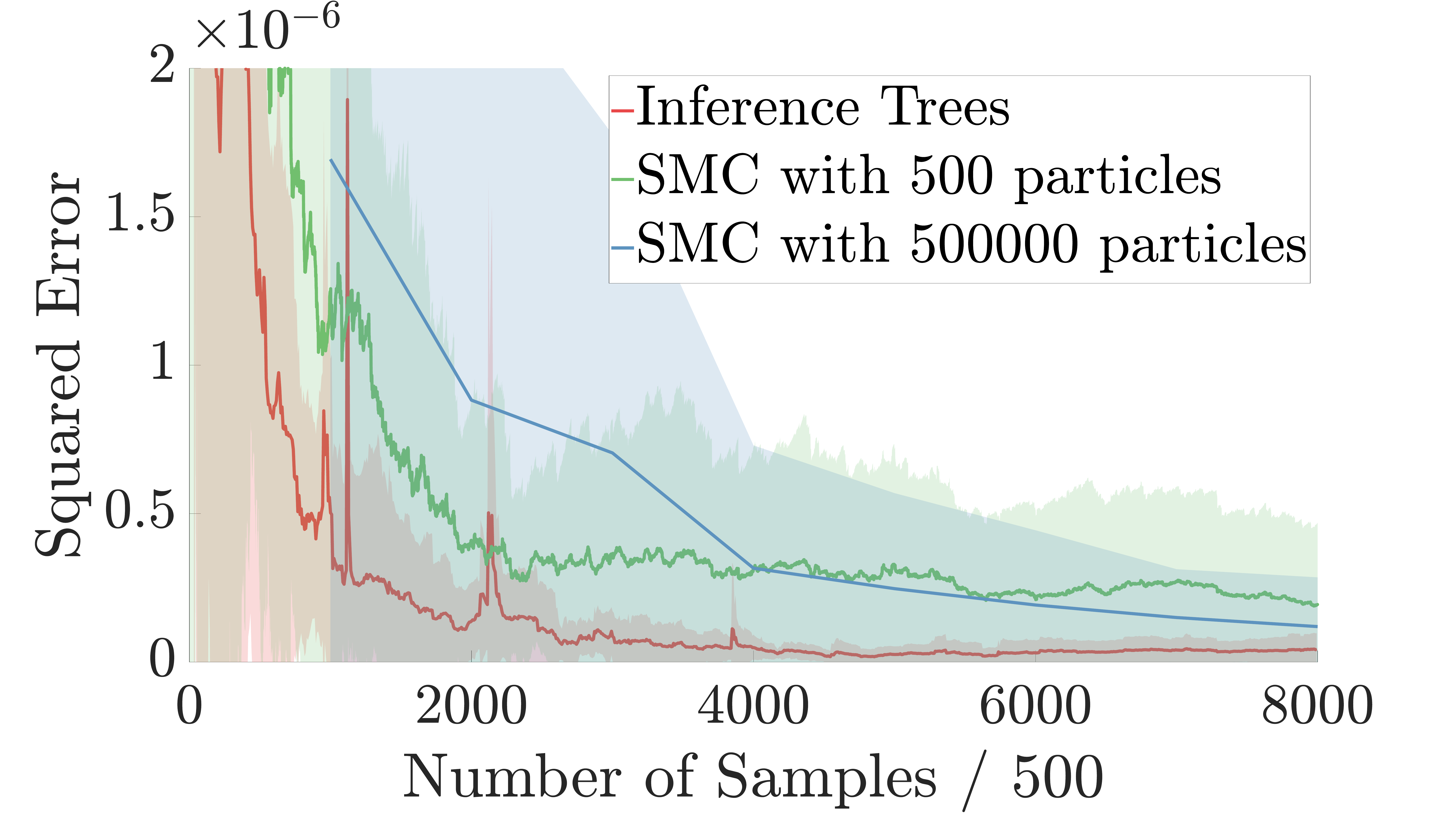}%
	\vspace{-6pt}
	\caption{Convergence of ITs on network model.  Solid
		lines show mean over 10 runs, with shaded region showing $\pm$ one standard
		deviation.	\vspace{-10pt}
		\label{fig:exp-plots}}
\end{wrapfigure}
 has weighted edges and we wish to estimate 
if the shortest path between two points exceeds a threshold. One possible application of such models would be in
modeling a traffic network, where the edges are streets connecting two points and the weights correspond to the commuting times on different edges
which are stochastic due to traffic levels and correlated because of the proximity of different streets to one another. 
We thus assume that there are noisy, correlated, observations for the edges weights, requiring inference,
while our threshold function means we are in a ``known $f$'' scenario, namely
we are estimating a form of tail integral.

The model is formally defined as
\begin{align}
x_{1:T} &\sim \mathcal{N} (x_{1:T} ; \mu, \Sigma) \\
y_{t} | x_{t}  &\sim \textsc{Student-T} \left(\frac{y_t-x_t}{\sigma} ; \nu \right) \quad \forall t \in \{1,\dots,T\}
\end{align}
where $x_t$ represents the unknown weights of edges, $y_t$ are noisy observations of those weights, 
and $\mu, \Sigma, \nu$, and $\sigma$ are known fixed parameters. Synthetic data was generated by setting $T=10$,
$\mu = [3,\dots,3]$, $\Sigma=I$, $\sigma=0.1$, and $\nu= 5$.  
We take the threshold as $3.8$ and look to estimate the probability that the shortest path exceeds this
threshold, which in our traffic analogy would correspond to not being able to reach a destination
on time.  We used SMC as the base inference with $500$ particles and used batches of $8$ runs
as per the chaos example.
Figure~\ref{fig:exp-plots} shows that IT outperform both SMC with the same number of
samples and SMC with $1000$ times more samples.

\newpage
\bibliography{refs}		
	
\end{document}